\documentclass{ws-ijfcs}
\usepackage{enumerate}
\usepackage{url}
\usepackage{fixltx2e}

\usepackage{amssymb}
\usepackage[figuresright]{rotating}
\usepackage{subcaption}

\usepackage[utf8]{inputenc}
\usepackage{complexity}
\usepackage{endnotes}
\usepackage{listings}
\usepackage{mathtools}
\usepackage{verbatim}
\usepackage{mathrsfs}
\usepackage{longtable}
\usepackage{enumitem}
\usepackage{etoolbox}
\usepackage{float}
\usepackage{color}
\usepackage{hyperref}

\usepackage{linenofb}

\usepackage{complexity}
\usepackage{comment}
\usepackage{longtable}

\newenvironment{code}
{\par%\runninglinenumbers
\modulolinenumbers[1]%
\linenumbersep-1em
\def\linenumberfont
{\normalfont\itshape}}

\renewcommand{\t}{\hspace*{3mm}}
\renewcommand{\tt}{\hspace*{3mm}}

\usepackage{mathtools}

\newcommand*{\DEBUG}{}%

\ifdefined\DEBUG
\newcommand{\fixme}[1]{{\textcolor{red}{\bf{\textsf{FIXME: #1}}}}}
\newcommand{\bug}[1]{{\textcolor{blue}{\bf{\textsf{BUG: #1}}}}}
\newcommand{\idea}[1]{{\textcolor{blue}{\bf{\textsf{IDEA: #1}}}}}

\newcommand{\TODO}[1]{{\textcolor{red}{\bf{\textsf{ 
TODO: #1
}}}}}
\else
\newcommand{\fixme}[1]{}
\newcommand{\bug}[1]{}
\newcommand{\TODO}[1]{}
\newcommand{\idea}[1]{}
\fi

\newclass{\COMSLIP}{COM\mbox{-}SLIP}
\newclass{\COMSLIPCUP}{COM\mbox{-}SLIP^{\cup}}

\newclass{\DCM}{DCM}
\newclass{\eDCM}{eDCM}
\newclass{\eNPDA}{eNPDA}
\newclass{\DPDA}{DPDA}
\newclass{\RDPDA}{RDPDA}
\newclass{\PDA}{PDA}
\newclass{\DCMNE}{DCM_{NE}}
\newclass{\TwoDCM}{2DCM}
\newclass{\NCM}{NCM}
\newclass{\sNCM}{sNCM}
\newclass{\eNCM}{eNCM}
\newclass{\eNQA}{eNQA}
\newclass{\eNSA}{eNSA}
\newclass{\eNPCM}{eNPCM}
\newclass{\mNPCM}{mNPCM}
\newclass{\eNQCM}{eNQCM}
\newclass{\eNSCM}{eNSCM}
\newclass{\DPCM}{DPCM}
\newclass{\NPCM}{NPCM}
\newclass{\NQCM}{NQCM}
\newclass{\NSCM}{NSCM}
\newclass{\NPDA}{NPDA}
\newclass{\TRE}{TRE}
\newclass{\NFA}{NFA}
\newclass{\DFA}{DFA}
\newclass{\NCA}{NCA}
\newclass{\DCA}{DCA}
\newclass{\DTM}{DTM}
\newclass{\NTM}{NTM}
\newclass{\DLOG}{DLOG}
\newclass{\CFG}{CFG}
\newclass{\ULGC}{ULGC}
\newclass{\CFGC}{CFGC}

\newclass{\ULG}{ULG}
\newclass{\ETOL}{ET0L}
\newclass{\EDTOL}{EDT0L}
\newclass{\CFP}{CFP}
\newclass{\ORDER}{O}
\newclass{\MATRIX}{M}
\newclass{\BD}{BD}
\newclass{\LB}{LB}
\newclass{\ALL}{ALL}
\newclass{\decLBD}{decLBD}
\newclass{\StLB}{StLB}
\newclass{\SBD}{SBD}
\newclass{\TCA}{TCA}
\newclass{\RNCSA}{RNCSA}
\newclass{\RDCSA}{RDCSA}

\newclass{\DCSA}{DCSA}
\newclass{\NCSA}{NCSA}
\newclass{\DCSACM}{DCSACM}
\newclass{\NCSACM}{NCSACM}
\newclass{\NTMCM}{NTMCM}

\newclass{\SMG}{SMG}
\newclass{\RLSMG}{RLSMG}
\newclass{\LSMG}{LSMG}

\newclass{\NTPCM}{NTPCM}
\newclass{\NCSPCM}{NCSPCM}

\newclass{\PBCM}{PBCM}
\newclass{\NTPBCM}{NTPBCM}

\newclass{\NTCM}{NTCM}

\newclass{\UFIN}{\LL(IND_{UFIN})}
\newclass{\UFINONE}{\LL(IND_{{UFIN}_1})}
\newclass{\FIN}{\LL(IND_{FIN})}
\newclass{\ILIN}{\LL(IND_{LIN})}
\newclass{\ETOLfin}{\LL(ET0L_{FIN})}

\newclass{\PTIME}{PTIME}

\newclass{\GSM}{GSM}

\newclass{\CNF}{CNF}

\newsavebox{\spacebox}
\begin{lrbox}{\spacebox}
\verb*! !
\end{lrbox}
\newcommand{\LL}{{\cal L}}
\newcommand{\MM}{{\cal M}}

\DeclareMathOperator{\yd}{yd}

\begin{document}%\def\docID{38} \input{../preamble.tex}

\markboth{Oscar H. Ibarra, Ian McQuillan}
{Language Acceptors with a Pushdown: Characterizations
and Complexity}

\catchline{}{}{}{}{}
%
%%%%%%%%%%%%%%%%%%%%%%%%%%%%%%%%%%%%%%%%%%%%%%%%%%%%%%%%%%%%%%%%%%%%

\title{Language Acceptors with a Pushdown: Characterizations
and Complexity\thanks{Electronic version of an article published as \href{https://doi.org/10.1142/S0129054124430044}{{\em International Journal of Foundations of Computer Science}, Vol. 36, No. 03, 2025, pp. 345--370} \copyright\ \href{https://www.worldscientific.com/worldscinet/ijfcs}{copyright World Scientific Publishing Company}.}
}

\author{Oscar H. Ibarra}

\address{Department of Computer Science\\ University of California, Santa Barbara, CA 93106, USA\\ \email{{ibarra@cs.ucsb.edu}}}

\author{Ian McQuillan}

\address{Department of Computer Science, University of Saskatchewan\\ Saskatoon, SK S7N 5C9, Canada\\ \email{mcquillan@cs.usask.ca}}

\maketitle

\begin{history}
\received{(Day Month Year)}
%\revised{(Day Month Year)}
\accepted{(Day Month Year)}
\comby{(xxxxxxxxxx)}
\end{history}

\begin{abstract}
We study one-way nondeterministic pushdown automata ($\NPDA$),
optionally with reversal-bounded counters.
Finite-turn pushdown automata are pushdown automata with a bound on the number of switches between pushing and popping.
We give new characterizations for 
finite-turn pushdown automata, and for finite-turn pushdown automata augmented
with reversal-bounded counters. 
The first is in terms of multi-tape nondeterministic finite automata ($\NFA$),
and the second is in terms of
multi-tape $\NFA$ with reversal-bounded counters.  We then use the
characterizations to determine the complexity of the languages defined
by these automata.  In particular, we show that languages
accepted by finite-turn $\NPDA$ augmented with reversal-bounded counters
are in $\NLOG$.  For the non-finite-turn case, the languages
are in $\DSPACE(\log^2 n)$ and in $\P$.  We also look at the
space complexity of languages accepted by two-way machines. In particular,
we show that every language accepted by a two-way $\NPDA$ with reversal-bounded
counters that makes a polynomial
(resp., exponential) number of input head reversals
is in $\DSPACE(\log^2n)$ (resp., $\DSPACE(n^2)$). This remains true if the pushdown can flip its contents a bounded number of times.
% and in $\DTIME(n^4\log^2 n)$.
\end{abstract}


\section{Introduction}
\label{sec:intro}

One-way nondeterministic pushdown automata, denoted by $\NPDA$, are one of the most well-studied classes of machines. 
This class of automata is quite practical to use. For example, it is well known that all languages accepted by pushdown automata are in $\P$ \cite{HU} and also in $\DSPACE(\log^2 n)$ (they are all accepted by $\log^2 n$ space-bounded deterministic Turing machines)
\cite{CFLslogsquared,HU1969}. Authors have also studied restrictions of $\NPDA$, such as $t$-turn $\NPDA$ (respectively finite-turn \NPDA), where in every accepting computation, the pushdown makes at most $t$ (respectively a finite number of) changes between non-decreasing and non-increasing the stack size \cite{GS2}.  
It is known that all languages accepted by finite-turn pushdown automata are in $\NSPACE( \log n) = \NLOG$ \cite{finiteturn} (they can be accepted in $\log n$ space by a nondeterministic Turing machine, which are all in $\P$).

Another type of machine model that has been extensively studied are one-way nondeterministic reversal-bounded multicounter machines. A one-way $k$-counter machine has a one-way read-only input tape, with $k$ counters as data store that can each hold a non-negative integer. At each step, a transition can detect whether a counter is positive or zero, and based on this status, either add one, subtract one, or keep it the same. While $2$-counter machines already have the same power as Turing machines \cite{HU}, we can restrict these counters to limit the power. The condition of being $r$-reversal-bounded (respectively reversal-bounded) enforces that in each accepting computation, the number of changes between sequences of non-decreasing transitions (adding one or zero at each step) and sequences of non-increasing transitions (subtracting one or zero at each step) on each counter is at most $r$ (respectively a finite number). It is also possible to combine different types of stores. The class of one-way reversal-bounded multicounter machines is denoted by $\NCM$, which were studied in \cite{Baker1974,Ibarra1978}. In \cite{Ibarra1978},
one-way nondeterministic  pushdown automata augmented with reversal-bounded counters were also studied, denoted here by $\NPCM$. 
$\NPCM$ is strictly more powerful than either $\NPDA$ or $\NCM$, and it has an $\NP$-complete non-emptiness problem (``given $M$, is $L(M) \neq \emptyset$?'')  \cite{HagueLin2011}. 
It was recently shown that $\NPCM$ has a decidable (in fact $\coNP$-complete) boundedness problem (``given $M \in \NPCM$, is $L(M)$ a bounded language?''  \cite{FoSSaCS}). 
Also recently, it was shown that the membership problem for two-way $\NPCM$ (where the input tape head can move in both directions) is $\NP$-complete \cite{decisionProblemsNCM}.

In this paper, we give new characterizations in terms of multi-tape acceptors.
We define a $\NPCM$ as being {\em finite-turn} if the pushdown is finite-turn.
The characterizations which are in terms of multi-tape
versions of nondeterministic finite automata ($\NFA$) and $\NCM$
are as follows:
\begin{enumerate}
\item finite-turn $\NPDA$ in terms of multi-tape $\NFA$,
\item finite-turn $\NPCM$ in terms of multi-tape $\NCM$.
\end{enumerate}
%These characterizations constrain usage patterns of the multi-tape input tapes. 

As by-product of the characterizations, we show that
all languages accepted by finite-turn $\NPCM$ are in $\NLOG$ and that
that all languages accepted by $\NPCM$ can be accepted by 
$\log^2(n)$ space-bounded deterministic Turing machines.
We also investigate the space complexity of languages accepted by
two-way machines.  In particular, we show that every language accepted by a
two-way $\NPCM$ that makes a polynomal (in the length of the input)
number of input head reversals is in $\DSPACE(\log^2n)$ (which generalizes to another machine model where the pushdown can flip its contents a bounded number of times).
This result  matches the space complexity of 
one-way $\NPDA$ languages (equal to the context-free languages) \cite{CFLslogsquared,HU1969}.
We also show that all $\NPCM$ languages
are in $\P$. For two-way $\NPCM$ languages, they are all in $\NP$ since the membership problem for two-way $\NPCM$ is $\NP$-complete \cite{decisionProblemsNCM}.
%; specifically, it is possible to generalize the CYK algorithm for context-free grammars so that we can decide if a given string is accepted in $O(n^4 \log^2 n)$ time.

The results in this paper significantly expand the knowledge regarding the time and space required to simulate generalizations and restrictions of one- and two-way pushdown automata with deterministic Turing machines.

\section{Preliminaries and Notation}

We assume introductory knowledge of automata and formal language theory \cite{HU}.

Let $\mathbb{Z}$ be the set of integers, $\mathbb{N}$  the positive integers, and $\mathbb{N}_0$ the non-negative integers. 
Given $x \in \mathbb{N}_0$, define $\pi(x) = 0$ if $x = 0$, and $\pi(x) = 1$ otherwise.
Given a set $X$ and 
$t \in \mathbb{N}$, let $\langle X \rangle^t$ be the set of all $t$-tuples over $X$, and given $x \in \langle X \rangle^t$, let
$x[j]$ be the $j$th component of $x$.
Given a finite alphabet $\Sigma$, let $\Sigma^*$ (respectively $\Sigma^+$) be the set of all words (respectively non-empty words) over $
\Sigma$. $\Sigma^*$ includes the empty word $\lambda$. A {\em language} $L$ is any subset of $\Sigma^*$, and a $t$-tuple
language $L$ is any subset of $\langle \Sigma^* \rangle^t$. 
Given a word $w = a_1 \cdots a_n, a_i \in \Sigma, 1 \le i \le n$, the {\em reversal}
of $w$, denoted $w^R$ is equal to $a_n a_{n-1} \cdots a_1$ with $\lambda^R = \lambda$. The
{\em length} of $w$, denoted by  $|w|$, is equal to the number of characters in $w$, and given $a\in \Sigma$, $|w|_a$ is the number of $a$'s in $w$.
%Given alphabet $\Sigma = \{a_1, \ldots, a_m \}$ and $w \in \Sigma^*$, the {\em Parikh image} of $w$, $\psi(w) = (|w|_{a_1}, \ldots, |w|_{a_m})$; and the Parikh image of a language $L \subseteq \Sigma^*$ is $\psi(L) = \{\psi(w) \mid w\in L\}$.
%Although we will not provide the formal definition of a language being semilinear, equivalently, a language is semilinear if and only if it has the same Parikh image as some regular language \cite{G78}.

Next, we will define automata with one pushdown plus some number of counters as stores. 
%All automata studied in this paper will be restrictions of this model.
%Multi-tape inputs 
%are somewhat less well studied than single tape variants, but 
%have been studied for $\NPDA$ \cite{multitapeNPDA}, $\NCM$ \cite{Ibarra1978}, and $\NPCM$ \cite{IbarraMcQuillanVerification}.
A {\em one-way nondeterministic pushdown $k$-counter machine} is a tuple $M = (Q,\Sigma,\Gamma, \delta,q_0,F)$
where $Q$ is a finite set of states, $\Sigma$ is the finite input alphabet, $\Gamma$ is the finite pushdown alphabet (that contains the fixed bottom-of-stack marker $Z_0$), $q_0 \in Q$ is the initial state, 
$F\subseteq Q$ is the set of final states, and $\delta$ is a relation that is a finite subset of $\Omega_1 \cup \Omega_2 \cup \cdots \cup \Omega_{k+1}$ where
\begin{eqnarray}
&&\Omega_1 = Q \times (\Sigma \cup \{\lambda\}) \times \{1\} \times \Gamma \times Q \times \Gamma^*,  \\
&& \Omega_i = Q \times (\Sigma \cup \{\lambda\}) \times \{i \} \times \{0,1\} \times Q \times \{ -1, 0, +1\} \mbox{~for~} 2 \le i \le k+1.
\end{eqnarray}
Essentially, this type of machine has $k+1$ stores, the first store is the pushdown, and the other $k$ stores are counters. The third component of each transition gives the store being read and changed, and each $\Omega_i$ gives all the possible transitions that can read and change the $i$th store. So, each transition applied only reads and changes one store at a time.
A {\em configuration} of $M$ is a tuple $(q,w,\theta_1 , \theta_2, \ldots, \theta_{k+1})$ where $q \in Q$ is the current state, 
$w \in \Sigma^*$  is the remainder of the input, $\theta_1 \in Z_0 (\Gamma-Z_0)^*$ is the pushdown contents, and $\theta_i \in \mathbb{N}_0$ is the current contents of counter $i -1$, for $2 \leq i \leq k+1$. 

We often associate labels from an alphabet $T$ bijectively to the transitions. 
Two configurations change as follows:
$$(q,aw, \theta_1 ,\theta_2,\ldots, \theta_{k+1}) \vdash^{\alpha} (q',w,\theta_1' ,\theta_2', \ldots, \theta_{k+1}'),$$
if there is a transition $\alpha$ in $\delta$ and $\Omega_i$ for some $i$ where $\theta_j = \theta_j' $ for all $j \ne i$, and 
if $i = 1$, $\alpha$ is $\delta(q,a,1,y, q', z), \theta_1 = \gamma y, \theta_1' = \gamma z$, and if
$i>1$, $\alpha$ is $\delta(q,a,i,\pi(\theta_i), q', z), \theta_i' = \theta_i + z$. 
Note that, because each counter contains a non-negative integer, it is not possible to subtract from a counter that contains zero.
We extend $\vdash^{\alpha}$ to $\vdash^{\alpha_1 \cdots \alpha_n}$ where $\alpha_1 \cdots \alpha_n \in T^*$ in the natural way to represent derivations of length zero or more. We can leave off the transition label to produce $\vdash$ (length one) and $\vdash^*$ (length zero or more).
%We also write $\vdash$ to represent $\vdash^{\alpha}$
%for some transition $\alpha$. We extend $\vdash^{\alpha}$ to words $\vdash^{\alpha_1 \cdots \alpha_n}$
%where $\alpha_1 \cdots \alpha_n \in T^*$ in the
%natural way to represent derivations of length zero or more.
%We usually denote an element 
%$q' \in \delta(q,a,i)$ by $\delta(q,a,i) \rightarrow q'$. 
An {\em accepting computation} on $w \in \Sigma^*$ is a sequence
\begin{equation} (q_0, w, Z_0, 0, \ldots, 0) \vdash \cdots \vdash (q_n, \lambda,\theta_1, \ldots, \theta_{k+1}), \label{computation} \end{equation} 
where $q_n \in F$. 
%The string ${\rm read}(\alpha_1 \cdots \alpha_n)$ is called the {\em tape order} of the computation (\ref{computation}).
The {\em language accepted} by $M$, $L(M) \subseteq \Sigma^*$ is the set of all $w \in \Sigma^*$ for which there is an accepting computation. %The {\em tape order language} $T(M)$ is the set of all tape orders of every 

Such a one-way nondeterministic pushdown $k$-counter machine $M = (Q,\Sigma,\Gamma, \delta,q_0,F)$ is 
{\it r-reversal-bounded} (respectively {\it reversal-bounded}) if in each accepting computation, the
number of changes between sequences of non-decreasing transitions and
sequences of non-increasing transitions (or vice versa) on each counter is at most $r$ (respectively a finite number).
Similarly, $M$ is $t$-turn (respectively finite-turn) if in each accepting computation, the number of changes between increasing the size and decreasing the size (or vice versa) of the pushdown is at most $t$ (respectively some number).

We will use several different classes of machines that are restrictions of one-way pushdown $k$-counter machines.
\begin{itemize}
\item $\NPDA$:  one pushdown but no counters,
\item finite-turn $\NPDA$: one finite-turn pushdown but no counters,
\item $\NCM$: no pushdown but some number of reversal-bounded counters,
\item $\NPCM$: a pushdown and some number of reversal-bounded counters,
\item finite-turn $\NPCM$: a finite-turn pushdown and some number of reversal-bounded counters.
\end{itemize}

It is known that for any machine with reversal-bounded counters, one can build an equivalent machine where all reversal-bounded counters are $1$-reversal-bounded \cite{Ibarra1978}, e.g.\ $\NPCM$ is the same as $\NPDA$ augmented by 1-reversal-bounded counters \cite{Ibarra1978}. 
It is also known that for any type of nondeterministic machine model with reversal-bounded counters, one can equivalently use
monotonic counters \cite{IbarraGrammars} instead of reversal-bounded counters. 
Such machines have an even number $k$
of counters that we identify by $C_1, D_1, \ldots, C_{k/2}, D_{k/2}$ that can only be
incremented but not decremented, transitions do not detect the counter status, 
and acceptance occurs when the machine 
enters an final state with counters $C_i$ and $D_i$
having the same value for each $i$.  Clearly, the monotonic counters
can be simulated by 1-reversal-bounded counters (at the end of the computation, $C_i$ and $D_i$ can be
decremented simultaneously and checked that they become zero at the
same time). For the converse, we can assume without loss of
generality that each 1-reversal-bounded counter $i$
is incremented and later decremented until it is zero.  Each such counter $i$ can be simulated
by two monotonic counters $C_i$ and $D_i$.  Increments in counter $i$ are done first
by incrementing $C_i$, then decrements on a non-zero counter $i$ are done by incrementing $D_i$,
and then finally it guesses that counter $i$ is zero and it simulates only transitions on counter $i$ being zero, and this guess is 
automatically verified at the end in order to accept. 
We will use the same notation as above
($\NPCM$, etc.) to also mean machines with monotonic counters.
Monotonic counters are often helpful in this paper because if we build a simulation of an accepting computation of a machine with another machine that applies the same counter changes but in a different order, then the resulting simulation will still have matching monotonic counters.

\vskip .25cm

\noindent
{\bf Convention}: For notational convenience, rather than calling
the monotonic counters $C_i$ and $D_i$, we would also just label
them as $C_1, C_2, \ldots, C_k$, where we assume that $k \ge 2$
is even, and acceptance requires that $C_i = C_{i+1}$ for odd $i$
at the end of the computation.

\vskip .25cm

Machines with a one-way input can be generalized to have multiple
one-way inputs. The input is a tuple $(w_1, \ldots, w_t) \in \langle \Sigma^* \rangle^t$, and transitions have an additional component $j \in \{1, \ldots, t\}$ where the input letter is read from the $j$th input tape. Configurations then have $t$ input words, and accepting computations finish with all $t$ input words being $\lambda$. In this case, $L(M) \subseteq \langle \Sigma^* \rangle^t$ which is a $t$-tuple language. For all classes of automata described so far, we preface the machine name with ``multi-tape'' to consider the multi-tape versions of these machines. Beyond multi-tape finite automata, multi-tape inputs 
%are somewhat less well studied than single tape variants, but 
have been previously studied for $\NPDA$ \cite{multitapeNPDA}, $\NCM$ \cite{Ibarra1978}, and $\NPCM$ \cite{IbarraMcQuillanVerification}.
We then define a homomorphism `${\rm read}$' from 
$T^*$ (an alphabet bijectively associated with the transitions) to $\{1, \ldots, t\}^*$ such that ${\rm read}(\alpha) = i$ when transition $\alpha \in T$  reads from input tape $i$.
The {\em tape order language} $T_{{\rm o}}(M)$ is the set of all tape orders of every 
accepting computation; so $T_{{\rm o}}(M) \subseteq \{1,\ldots, t\}^*$.

A context-free grammar (abbreviated $\CFG$) is a tuple $G = (V, \Sigma, P, S )$, where $V$ and $\Sigma$ are finite and disjoint alphabets of nonterminals, and terminals, respectively, $S \in V$ is the starting nonterminal, and $P$ is a finite set of productions of the form $A \rightarrow w, A \in V, w \in (V \cup \Sigma)^*$. 
We will define language generation in two ways. The first is with the derivation relation. We say 
$\alpha A \beta \Rightarrow \alpha w \beta$ if $\alpha,\beta \in (V\cup \Sigma)^*$ and $A\rightarrow w \in P$. We let
$\Rightarrow^*$ be the reflexive and transitive closure of $\Rightarrow$. Any string $\alpha$ such that $S \Rightarrow^* \alpha$ is called a sentential form, and the language generated by $G$, $L(G) = \{w \mid S \Rightarrow^* w \in \Sigma^*\}$.
It is well known that we can equivalently define language acceptance in terms of leftmost derivations where we always force the leftmost nonterminal to be rewritten, and also by complete derivation trees \cite{HU}. 
A complete derivation tree is a directed
tree where each vertex is labelled by an element of $V\cup\Sigma \cup \{\lambda\}$, the root is labelled by $S$, leaves are labelled by elements of $\Sigma \cup \{\lambda\}$, and if an inner-vertex is labelled by $A$ and its children are labelled by $u_1, \ldots, u_n$ from left-to-right, then $A \rightarrow u_1 \cdots u_n \in P$. The set of complete derivation trees is denoted by $T(G)$. The yield of a tree $t \in T(G)$, $\yd(t)$, is the concatenation of labels on leaves obtained by a preorder traversal. It is known that $ \{\yd(t) \mid t \in  T(G)\} = L(G)$ and thus derivations and the use of complete derivation trees are equivalent.
 Any language generated by a context-free grammar is a context-free language. 
It is well known that languages generated by $\CFG$s are equal to those accepted by $\NPDA$s \cite{HU}. 

A context-free grammar with $k \ge 0$ monotonic counters \cite{IbarraGrammars} is a tuple $G = (V,\Sigma,P,S)$ where $k$ is even, 
 $V$ and $\Sigma$ are the finite and disjoint alphabets of nonterminals and terminals, respectively, $S \in V$ is the starting nonterminal, and P is a finite set of productions of the form
$A \rightarrow (w, i_1, \ldots, i_k)$, where $A \in V, w \in (V \cup \Sigma)^*$, and $i_j \in \mathbb{N}_0$ for all $1 \le j \le k$.
Given $t \in T(G)$ where $T(G)$ is the set of normal context-free complete derivation trees, we say $t$ {\em satisfies counters} if
the sum of the counters of productions used in $t$ is $(v_1, \ldots, v_k)$ and $v_i = v_{i+1}$ for all $i$ odd, $1 \le i \le k-1$.
Let $T_{{\rm e}}(G)$ be the set of all $t \in T(G)$ that satisfies counters. Further, $L(G) = \{\yd(t) \mid t \in T_{{\rm e}}(G)\}$.
We denote the class of context-free grammars with monotonic counters by $\CFGC$.
Similarly to context-free grammars, we can equivalently define language acceptance using derivations where now configurations are a tuple $(\alpha ,v_1, \ldots, v_k),  \alpha \in (V \cup \Sigma)^*, v_i \in \mathbb{N}_0$; the derivation relation is defined as $(\alpha A \beta,v_1,\ldots, v_k) \Rightarrow (\alpha w \beta, x_1,\ldots, x_k)$ where $(A \rightarrow w, x_1-v_1, \ldots, x_k-v_k) \in P$; each configuration $(\alpha, v_1, \ldots, v_k)$ where $(S, 0, \ldots, 0) \Rightarrow^* (\alpha, v_1, \ldots, v_k)$ is a sentential form; and the language generated by $G$ is equal to $\{w \mid S \Rightarrow^* (w,v_1, \ldots, v_k), w  \in \Sigma^*, v_i = v_{i+1} \mbox{~for~} i \mbox{~odd}\}$ which coincides with the notion above using derivation trees.
Given a $\CFGC$ $G$, the {\em underlying context-free grammar} is the context-free grammar obtained from $G$ by leaving off the counter changes from each production. It is known that the languages generated by $\CFGC$ are equal to the language accepted by $\NPCM$ \cite{IbarraGrammars}.

Finally, let $\DTM$ (respectively $\NTM$) be the class of deterministic (respectively nondeterministic Turing machines with a two-way read-only input tape and multiple read/write worktapes.

\section{Characterization of Finite-Turn $\NPDA$ and Finite-Turn $\NPCM$}
\label{sec:characterizations}

In this section, we give a characterization
of finite-turn $\NPCM$ (respectively, finite-turn $\NPDA$) in terms of multi-tape $\NCM$
(respectively, multi-tape $\NFA$). Previous characterizations of
finite-turn $\NPDA$ are in terms of grammars
\cite{GS2} and not machines (although we use this characterization in terms of grammars to help prove our characterization). There has not previously
been any characterization of finite-turn $\NPCM$. %Using machine models instead of grammars allows
%the proofs and algorithms to be more easily adapted to other machine models, as we do in this paper.

A $\CFG$ (context-free grammar) $G = (V, \Sigma, P, S )$ is $n$-{\em ultralinear} (or simply ultralinear) \cite{GS2} if there is a partition
$\{V_0, V_1, \ldots,  V_n \}$ of $V$ such that $S$ is in $V_n$ and if $A$ is in 
$V_i, 0 \le i \le n$, then $A \rightarrow  w$ is in
$P$ implies either $w \in \Sigma^* V_i \Sigma^*$ or $w \in (\Sigma \cup V_0 \cup \cdots \cup V_{i-1})^*$.
The language generated by the grammar is called an
($n$-)ultralinear language.  If $n = 0$, then the grammar (language) 
is called a linear grammar (language). The class of ultralinear grammars is denoted by $\ULG$.

The following proposition was shown by Ginsburg and Spanier \cite{GS2}:
\begin{proposition} \cite{GS2} \label{propB1}
We can effectively construct, given a finite-turn $\NPDA$ $M$, an ultralinear grammar
$G$ equivalent to $M$, and conversely.
\end{proposition} 

Here we define a $\CFGC$ to be $n$-ultralinear (respectively ultralinear) if its underlying context-free grammar is $n$-ultralinear (respectively ultralinear). The class of ultralinear grammars with monotonic-counters is denoted by $\ULGC$.

\begin{comment}

Context-free grammars (respectively ultralinear) can be augmented with
monotonic counters.
With this model, there are an even number $k \ge 2$ of counters, and
the rules in $P$ are now of the form 
$A \rightarrow (w, i_1, \ldots, i_k)$, where each $i_j \in \mathbb{N}_0$ (respectively where
$A \rightarrow w$ is an ultralinear rule).
Given $t \in T(G)$ where $T(G)$ is the set of normal context-free complete derivation trees, we say $t$ {\em satisfies counters} if
the sum of the counters of productions used in $t$ is $(v_1, \ldots, v_k)$ and $v_i = v_{i+1}$ for all $i$ odd, $1 \le i \le k-1$.
Let $T_{{\rm e}}(G)$ be the set of all $t \in T(G)$ that satisfies counters. Further, $L(G) = \{\yd(t) \mid t \in T_{{\rm e}}(G)\}$.
We call these context-free grammars (respectively $n$-ultralinear grammars) with monotonic counters, and we denote the class of context-free grammars (respectively ultralinear grammars) with monotonic counters by $\CFGC$ (respectively $\ULGC$). 
\end{comment}

Proposition \ref{propB1} can be generalized:

\begin{proposition} \label{propB2}
We can effectively construct, given a finite-turn $\NPCM$ $M$, an $\ULGC$
$G$ equivalent to $M$, and conversely.
\end{proposition}
\begin{proof}
Let $M$ be a finite-turn $\NPCM$ with input alphabet $\Sigma$
which has $k$ monotonic counters.  Let $a_1, \ldots, a_k$
be new symbols, where each $a_i$ is associated with counter $i$.
Let $\Delta = \Sigma \cup \{a_1, \ldots, a_k\}$.
We construct a finite-turn $\NPDA$ $M'$ which on input $w$ in $\Delta^*$, simulates
$M$ as follows:
Any transition that uses the pushdown or that adds zero to a counter
is simulated verbatim. Any transition that reads $a \in \Sigma \cup \{\lambda\}$ and increments counter $j$ by one,
$1 \le j \le k$, is simulated by reading $a$ followed immediately by $a_j$
on the input. Of course, $M'$ does not check that the counter values match.
By using Proposition \ref{propB1}, we can construct from $\NPDA$ $M'$ an
$\ULG$ $G'$ such that $L(G') = L(M')$.  Finally,
from $G'$ we construct an $\ULGC$ $G$ by converting each rule
of the form $A \rightarrow w$ in $G'$ to a rule
$A \rightarrow (z, i_1, \ldots, i_k)$, where $z$ is obtained from $w$
with all the symbols in $\{a_1, \ldots, a_k\}$ deleted, and
$i_j = |w|_{a_j}$, $1 \le j \le k$.
Since this enforces that counter values are matching, $L(G) = L(M)$.

For the converse, let $G$ be an $\ULGC$ with $k$ counters, and let $a_1, \ldots, a_k$ be new symbols. 
We construct from $G$
an $\ULG$ $G'$ as follows:  If $A \rightarrow (w, i_1, \ldots, i_k)$
is a rule in $G$, then $A \rightarrow w  a_1^{\vert i_1 \vert} \ldots a_k^{\vert i_k \vert}$ is a rule in $G'$.
Next, we construct a finite-turn $\NPDA$ $M'$
from $G'$ such that $L(M') = L(G')$ using Proposition \ref{propB1}.  Finally,
construct from $M'$ a finite-turn $\NPCM$ $M$ (with $k$ monotonic counters) which simulates $M'$ as follows:
$M$ simulates transitions of $M'$ that read letters of $\Sigma\cup \{\lambda\}$, but instead of reading letters 
$a_j$ (for $1 \le j \le k$), it reads $\lambda$ and adds one to counter $j$. It is clear that $L(G) = L(M)$.
\end{proof}

We now relate $\ULGC$ to multi-tape $\NCM$. 
Each complete derivation tree of $G$ can be represented with a well-formed bracketed string
over the letters $\texttt{[}$ and $\texttt{]}$ describing only the patterns of how the nonterminal partitions are
changing, as follows:
Given an $n$-ultralinear grammar (with or without counters)  $G = (V,\Sigma,P,S)$, 
let $G' = (V',\Sigma \cup \{\texttt{[},\texttt{]}\}, P', S')$ be another $n$-ultralinear grammar where $S'$ is a new symbol,
$S' \rightarrow {\texttt{[}} S {\texttt{]}} \in P'$ with zero added to the counters, and
(counters are preserved in the rest of the construction), for each production $A \rightarrow u_0 A_1 u_1 \cdots u_{t-1} A_t u_t  \in P$, where
$t \ge 0, A \in V_i,  A_1, \ldots, A_t \in V_1 \cup \cdots \cup V_{i-1}, u_0, \ldots, u_t \in \Sigma^*$, replace it with
$A\rightarrow u_0  {\texttt{[}}  A_1  {\texttt{]}} u_1 \cdots u_{t-1}  {\texttt{[}}  A_t  {\texttt{]}} u_t  \in P'$, and all other productions in $P$ of the form $A \rightarrow u B v, A,B \in V_i, u,v \in \Sigma^*$ are kept in $P'$. Indeed, every derivation
tree in $G$ has a corresponding tree in $G'$ where letters of $\{\texttt{[} , \texttt{]}\}$ are included.
The {\em well-formed bracketed pattern} of a complete derivation tree of $G$ is the string obtained by projecting the yield of the corresponding tree in $G'$
to $\{\texttt{[} , \texttt{]}\}^*$. The set of all well-formed bracketed patterns of complete derivation trees in $T(G)$ (if it has counters, in $T_{{\rm e}}(G)$) is denoted by $P(G)$.
Certainly, for any ultralinear grammar, $P(G)$ is finite. Below we will show that every language $L$ generated by an
ultralinear grammar with or without counters can also be generated by a grammar $G$ where $P(G)$ contains at most one pattern (it is empty if and only if $L = \emptyset$).
Any well-formed bracketed pattern with $m$ matching bracket pairs is called an $m$-pattern.

For example, consider a $1$-ultralinear
grammar $G = (V,\Sigma,P,S)$ (with no counters), where 
the set of nonterminals is $V = V_0 \cup V_1$, the start nonterminal
is in $V_1$, and the rules are of the form:
\begin{enumerate}
\item $A \rightarrow uBv$, where $A,B \in V_i, i \in \{ 0, 1\}$, $u,v \in \Sigma^*$,
\item $A \rightarrow u$, where $A \in V_i, i \in \{ 0, 1\}$, $u \in \Sigma^*$,
\item $A \rightarrow v_0 A_1 v_1 A_2 v_2 \cdots v_{t-1} A_t v_t$, 
$A \in V_1, A_1, \ldots, A_t \in V_0$,
$v_0, v_1, \ldots, v_t \in \Sigma^*$.
\end{enumerate}
If a derivation only uses a sequence of productions of
type 1 with $V_1$, then one of type 2, then it has bracketed pattern $\texttt{[ ]}$ from the application of the initial production $S' \rightarrow {\texttt{[}} S {\texttt{]}}$.
If a derivation starts with productions of type 1 with $V_1$, followed by
one of type 3, with each branch then following a sequence using $V_0$ followed by a transition
of type 2, this  has pattern
$\texttt{[ } \overbrace{\texttt{[ ] [ ]} \cdots \texttt{[ ] }}^{t \mbox{ sections}} \texttt{ ]}$; the outer brackets come from the application of $S' \rightarrow {\texttt{[}} S {\texttt{]}}$ and the inner brackets come from the applied production of type 3. So both strings are in $P(G)$.
Observe that a new pair of matching brackets is added along a path from the root to a leaf at the beginning, and when the nonterminal set changes using a production of type 3.

However, every ultralinear grammar (with or without counters) can be converted to one where every complete derivation tree has the same pattern.
Let $G = (V,\Sigma,P,S)$ be an $n$-ultralinear grammar with $V$ partitioned into $V_0, \ldots, V_n$. We say $G$ is in $t$-{\em normal form} for $t \ge 1$
if all rules 
are of the form
\begin{enumerate}
\item $A \rightarrow uBv$, where $A,B \in V_i, 0 \le i \le n$, $u,v \in \Sigma^*$,
\item $A \rightarrow u$, where $A \in V_0,  u \in \Sigma^*$,
\item $A \rightarrow A_1 A_2 \cdots A_t $, 
$A \in V_{i+1}, A_1, \ldots, A_t \in V_i, 0 \le i \le n-1$, where $t$ is the same for all such rules in the grammar.
\end{enumerate}

\begin{lemma} \label{tonormal}
Given each $n$-ultralinear grammar $G= (V,\Sigma,P,S)$ (respectively with monotonic counters), another $n$-ultralinear grammar (respectively with  monotonic counters) $G'$ can be built
in $t$-normal form for some $t \ge 1$ such that $L(G)= L(G')$. 
\end{lemma}
\begin{proof}
We do this conversion in two steps. In describing the construction, the counters are kept the same for any production that is either kept or modified from a production in the original grammar, and new productions all add zeros to the counters.

First, we describe an intermediate construction without the condition that $t$ must be the same for all rules of type 3.
We create $\bar{G} = (\bar{V} , \Sigma,\bar{P},S)$ with $\bar{V}$ partitioned into $\bar{V}_0, \ldots, \bar{V}_n$ as follows. First, all nonterminals in each $V_i$, $1 \le i \le n$, are in $\bar{V}_i$. All productions of the form
$A\rightarrow u B v, A,B \in V_i, u,v \in \Sigma^*$ and $A \rightarrow u, A \in V_0, u \in \Sigma^*$ are kept verbatim.
For each $Z \in V_i$ and $i, 0 \le i <n$, introduce new nonterminals $Z(i) \in \bar{V}_i$ and  production $Z(i) \rightarrow Z$; for each $j$, $i < j \le n-1$ introduce $Z(j) \in \bar{V}_j$; and for each $l$, 
$i \le l < n -1 $, introduce production 
$Z(l+1) \rightarrow Z(l)$.
Essentially, if we use a nonterminal $Z(j) \in \bar{V}_j$ where $Z \in V_i, i \le j$, it must switch to nonterminal $Z(j-1) \in \bar{V}_{j-1}$, then $Z(j-2) \in \bar{V}_{j-2}$, etc.\
one at a time until $Z(i) \in \bar{V}_{i}$ where it must go to the original nonterminal $Z \in V_i$.

Also, for all productions $Y$ of the form $A \rightarrow u_0 A_1 u_1 \cdots A_r u_r, r \ge 1$ where $A \in V_i$, and all 
$A_j \in V_0 \cup \cdots \cup V_{i-1}$, replace it with 
$A\rightarrow Y_1 \cdots Y_r$ where $Y_1, \ldots, Y_r$ are new nonterminals in $\bar{V}_{i-1}$ associated with $Y$, and also introduce $Y_1 \rightarrow u_0 A_1(i-1) u_1, Y_2 \rightarrow A_2(i-1)u_2, \ldots, Y_r \rightarrow A_r(i-1) u_r$. 
For all productions $Y$ of the form $A \rightarrow u$, $A \in V_i, i \ge 1, u \in \Sigma^*$, replace it with
$A \rightarrow Y_{i-1}, Y_{i-1} \rightarrow  Y_{i-2}, \ldots, Y_0 \rightarrow u$ where $Y_j$ is another new nonterminal in $\bar{V}_j$ associated with the production $Y$, $0 \le j \le i-1$.
In this way, it is enforced that as it moves down the tree, it must decrease the nonterminal partition one at a time. It is evident that $L(G) = L(\bar{G})$.

Next, we convert $\bar{G}=(\bar{V},\Sigma,\bar{P},S)$ into $G' = (V', \Sigma, P', S)$ with $V'$ partitioned into $V_0', \ldots, V_n'$, where we enforce that $t$ must be the same for all rules.
Let $t$ be the maximum of the number of nonterminals on the right hand side of any rule in $\bar{P}$.
Let $X_i \in V_i'$ for $0  \le i \le n$ be new nonterminals, and create new productions
$X_i \rightarrow X_{i-1}^t$  for $i>0$ and $X_i \rightarrow \lambda$ for $i = 0$.
Also, for each production in $P$ of the form
$A \rightarrow A_1 \cdots A_r $, $ A \in V_i, i \ge 1, A_1, \ldots, A_r \in V_{i-1}, 1 \leq r \leq t$, we replace it with
$A \rightarrow A_1 \cdots A_r X_{i-1}^{t-r}$. All other productions are kept.
 
It is clear that $L(\bar{G}) = L(G') = L(G)$, and that $G'$ is $n$-ultralinear and in $t$-normal form.
\end{proof}

Next we see that only zero or one well-formed bracketed pattern is needed.
For $t \ge 1$, let $p_0 = p_{0,t} = \texttt{[ ]}$, and for $i \ge 1$, $p_{i,t} = \texttt{[} (p_{i-1, t})^t \texttt{]}$.
\begin{lemma} \label{dertree}
Given any $n$-ultralinear grammar $G$ (with or without monotonic counters) in $t$-normal form, $t \ge 1$, then 
$P(G) \subseteq \{p_{n,t}\}$, and they are equal if and only if $L(G)$ is non-empty.
\end{lemma}
\begin{proof}
We will deliberately prove that $P(G) \subseteq \{p_{n,t} \}$ where $G$ does not have counters, as any grammar $G$ with counters
satisfies $T_{{\rm e}}(G) \subseteq T(G)$.
 
We will prove this by induction. 

First, for the base case, it is clear that any $0$-ultralinear grammar $G$ has $P(G) \subseteq \{p_0\}$.

Let $k \ge 0$, and assume by induction that for any $k$-ultralinear grammar $G'$ in $t$-normal form has
$P(G') \subseteq \{p_{k,t}\}$. Consider a $k+1$-ultralinear grammar $G= (V,\Sigma,P,S)$ with $V$ partitioned into 
$V_0, \ldots, V_{k+1}$.

Given any $A \in V_k$, let
$G(A)$ be the $k$-ultralinear grammar obtained from $G$ by setting $A$ as starting nonterminal, only including
nonterminals in $V_0 \cup \cdots \cup V_k$, and only keeping productions that use these nonterminals.
By the inductive hypothesis, for all such $A$, $P(G(A)) \subseteq \{p_{k,t}\}$.
In every complete derivation tree of $G$, it starts with a nonterminal in $V_{k+1}$ and only uses productions of the form
$A \rightarrow uBv, A,B\in V_{k+1}, u,v \in \Sigma^*$, until it uses a single production of the form
$A\rightarrow A_1 \cdots A_t, A_1, \ldots, A_t \in V_k$. For each subtree rooted at each $A_i$, it is a complete 
derivation tree in $G(A_i)$ and so $P(G(A_i)) = \{p_{k,t}\}$. Hence every complete derivation tree in $G$ has pattern
$\texttt{[}$ followed by $t$ copies of $p_{k,t}$, followed by $\texttt{]}$. 
\end{proof}

\begin{comment}
As $G$ is ultralinear, clearly the set of well-formed bracketed patterns representing
complete derivation trees of $G$ in this way is finite.
Given an $n$-ultralinear grammar $G$, we can effectively
derive an (upper-bound) $m$ such that every bracketed pattern corresponding to every complete
derivation tree is 
at most of length $2m$. In the example above, $m = t+1$. 
We call this an $m$-pattern.
In what follows, $P(m)$ will denote a set of well-formed $m$-patterns.  
\end{comment}

Given an $m$-pattern $p$ and a $2m$-tuple 
$z = (x_1,y_1, \ldots, x_m, y_m), x_i,y_i \in \Sigma^*$, let $z^{(p)}$ be the string obtained
from $p$ by replacing the $i$th $\texttt{[}$ by $x_i$ and its corresponding matching $\texttt{]}$
by $y_i^R$.
If $L \subseteq \langle \Sigma^* \rangle^{2m}$, 
let $L^{(p)}  = \{z^{(p)} ~|~ z \in L \}$.

% We extend this to sets of $2m$-tuples $Z$ to produce $Z^p$ in the natural way.

A $2m$-tape $\NCM$ $M$ with $2m$ tapes whose input tapes we will label
respectively $T_1, R_1, \ldots , T_m, R_m$, is called {\em paired} 
%
%if $T(M_p) \subseteq (T_1 + R_1)^* \cdots (T_m + R_m)^*$; that is,
%
if $T_1$ and $R_1$ are read first, then $T_2$ and $R_2$ are read next, 
and so on, with $T_m$ and $R_m$ read last. That is, $T_{{\rm o}}(M) \subseteq (T_1+R_1)^* \cdots (T_m + R_m)^*$. Below
we will determine that every finite-turn $\NPCM$ language $L$ has a paired $2m$-tape $\NCM$ and a well-formed $m$-pattern
$p$ with $L = L(M)^{(p)}$, and vice versa.
The next lemma associates each ``linear section'' of complete derivation trees
with two paired tapes in such a way that each reads
terminals produced on one side of the nonterminal chain. 

%If $M$ is a paired $2m$-tape $\NCM$ ($\NFA$) and $p$ is an $m$-pattern,
%$L(M)^p$ is called a  paired $2m$-tape  $\NCM$ ($\NFA$)
%$m$-pattern language.

\begin{lemma} \label{corB1}
If $M$ is a finite-turn $\NPDA$ (respectively  finite-turn $\NPCM$), then there is a paired $2m$-tape $\NFA$ (respectively paired $2m$-tape $\NCM$) $M'$ and an $m$-pattern $p$
such that $L(M) = L(M')^{(p)}$.
\end{lemma}

\begin{proof} We will show the result for finite-turn $\NPCM$ (with monotonic counters), and it will be clear that it also works without counters.

Let $M$ be a finite-turn $\NPCM$ and let $G$ be an $n$-ultralinear $\ULGC$ equivalent to $M$ by Proposition \ref{propB2}. We assume that
$G$ is in $t$-normal form, for some $t \ge 1$ by Lemma \ref{tonormal}.
By Lemma \ref{dertree}, $P(G) \subseteq \{p_{n,t}\}$. Let $p = p_{n,t}$, where $p$ is an $m$-pattern.

We construct a $2m$-tape $\NCM$ $M'$ that
simulates a leftmost derivation of $G$. This construction is similar to the standard simulation of leftmost derivations of context-free grammars by pushdown automata \cite{HU}, but here we argue that a finite pushdown is all that is needed, which can be stored in the state. We build $M'$ so that it has a stack of size $m$. 
In the simulation described, the productions are simulated by transitions that add the same counter values to the counters.
As linear rules from and to nonterminals in $V_n$ are used of the form $A \rightarrow u B v$, $u,v \in \Sigma \cup \{\lambda\}, A,B \in V_n$, $M_i$ keeps track of the current nonterminal of the leftmost derivation at the top of the finite stack (always replacing the topmost symbol with the new nonterminal in this section), and  $M'$
uses input tapes $T_1$ and $R_1$ in parallel to read $u$ from $T_1$ and $v^R$ from $R_1$.
Then for this sequence of linear productions using only nonterminals of $V_n$, $S \Rightarrow^* x_1 B y_1$
for some $x_1,y_1 \in \Sigma^*, B \in V_n$, $T_1$ reads $x_1$ while $R_1$ reads $y_1^R$, as the letters
to the right of the nonterminal are read in reverse from the order they appear in the derived word. This is the part of the words derived in complete derivation trees 
corresponding to the first $\texttt{[}$ and the last $\texttt{]}$ in $p$. 
Next, a production of type 
$A \rightarrow A_1 A_2 \cdots A_t$, $A \in V_n, A_1, \ldots, A_t \in V_{n-1}$,
is used, and is simulated by replacing $A$ on the top of the stack by $(A_1 \cdots A_t)^R$, and it continues simulating from $A_1$ by replacing $A_1$ on the top of the stack with those that can be obtained using linear productions
starting from $A_1 \in V_{n-1}$. This part of the derivation is simulated similarly using
$T_2$ and $R_2$, and it corresponds to the second $\texttt{[}$ and its matching $\texttt{]}$. The leftmost derivation continues in this way. Since each of these $t$ nonterminals have one matching set of brackets in the well-formed bracketed pattern, a stack of size $m$ is all that is needed for this simulation.

It is evident that continuing in this fashion, $M'$ can simulate $M$ and based on the order
in which terminal letters are produced in $p$ versus the order they are read by $M$,
that $L(M) = L(M')^{(p)}$
and that $M'$ is paired.
%. Furthermore, the tape order language is a subset of 
%$(T_1 + R_1)^* \cdots (T_m + R_m)^*$ and therefore $M'$
\end{proof}

We now show the reverse direction.

\begin{lemma}
Given a paired $2m$-tape $\NFA$ (respectively $\NCM$) for some $m$, and a well-formed $m$-pattern $p$, there is a finite-turn
$\NPDA$ (respectively $\NPCM$) $M'$ such that $L(M') = L(M)^{(p)}$.
\end{lemma}
\begin{proof}
Given a paired $2m$-tape $\NCM$ $M = (Q,\Sigma,\delta,q_0,F)$ and an $m$-pattern $p$, we
construct a finite-turn 
$\NPCM$ $M' = (Q',\Sigma, \Gamma,\delta',q_0',F')$ that accepts $L(M)^{(p)}$. Intuitively, $M'$ operates as follows
when given input $w \in \Sigma^*$:
$M'$ simulates the reading of the
tapes according to pattern $p$.
We give an example before describing the general simulation. On input $w = x_1 x_2 y_2^R x_3 x_4 y_4^R y_3^R y_1^R$
and pattern $p = \texttt{[ [ ] [ [ ] ] ]}$, $M'$ simulates the
computation of $M$
%on $(x_1, x_2, x_3, x_4, x_4, x_5)$
(while guessing the decomposition of $w$) as follows:
\begin{itemize}
\item $M'$ first guesses a list of states in $Q$, $(q_0, q_1, \ldots, q_4)$
where for $0 \le i < 4$, $(q_i, q_{i+1})$ means $M$ when started in state $q_i$ ends in state $q_{i+1}$ after reading from (paired) tapes $T_i$ containing $x_{i+1}$ and $R_i$ containing $y_{i+1}$ in parallel, with $q_0$ initial state of 
$M$ and $q_4 \in F$. Then:
\item 
$M'$ reads $x_1$ and pushes it in the stack.  
\item
$M'$ reads $x_2$ and pushes it in the stack. Then starting in 
   state $q_2$ simulates the computation of $M$ on $x_2^R$ and $y_2^R$ ``backwards'' where it reads $x_2^R$ by popping the stack in parallel to reading $y_2^R$ on the input, and confirms that the state entered after the process is $q_1$.
\item
$M'$ reads $x_3$ and pushes it in the stack. 
\item 
$M'$ reads $x_4$ and pushes it in the stack. Then starting in 
   state $q_4$ simulates the computation of $M$ on $x_4$ and $y_4$ backwards by reading $x_4^R$ from the stack and $y_4^R$ from the input and confirms that the state entered after the process is $q_3$.
\item
$M'$ starting in state $q_3$ simulates the computation of $M$ on $x_3$ and $y_3$ backwards on $x_3^R$ from the stack and 
$y_3^R$ from the input and confirms that the state entered after the process is $q_2$.
\item
$M'$ starting in state $q_1$ simulates the computation of $M$ on $x_1$ and $y_1$ backwards on
both $x_1^R$ and $y_1^R$ and confirms that the state entered after the process is $q_0$. Then $M'$ accepts.  
\end{itemize}

More generally, $M'$ starts by guessing $m+1$ states of $M$, $(q_0, \ldots, q_m)$ where $q_0$ is initial in $M$ and $q_m \in F$,
and $M'$ will simulate an accepting computation of $M$ on input $x = (x_1,y_1, \ldots, x_m, y_m)$. The paired $M$ starts by reading from the first two tapes $T_1,R_1$
between $q_0$ and $q_1$, then the next two $T_2, R_2$ between $q_1$ and $q_2$, etc., until it reads from the last two tapes 
$T_m,R_m$ between $q_{m-1}$ and $q_m$.

On input $w$ from left-to-right, $M'$ nondeterministically guesses a decomposition into $2m$ segments $r_1 \cdots r_{2m}$.
If the $i$th  $\texttt{[}$ of $p$ is at position $j$, let $x_i = r_j$, and if the $i$th  $\texttt{]}$ of $p$ is at position $j$, let
$z_i = r_j$ and $y_i = z_i^R$. Then $M'$ will simulate $M$ on $x^{(p)}$ as follows:
For each $i$ from $1$ to $m$, $M'$ reads and pushes $\$x_i$ to the pushdown ($\$$ a new symbol), and when it reaches $z_j = y_i^R$ on the input, it can read $y_i^R$ from the input in parallel to reading $x_i^R$ from the pushdown as it pops from the pushdown (which must occur since $p$ has well-formed brackets). Notice that any machine with monotonic counters can be simulated ``in reverse'' by switching the start and end state in each transition. This works because the transitions applied do not depend on the counter status, and the counter status only affects acceptance by only considering computations where counters are matching. Indeed, even taking an accepting computation of $M$ and reversing segments of the computation will still result in the counters being matched. Hence, $M'$ can simulate $M$ ``in reverse'' for the portion of the computation reading $x_i$ and $y_i$ from tapes $T_i$ and $R_i$ respectively between $q_{i-1}$ and $q_i$ ``in reverse'' by reading $x_i$ in reverse from the pushdown and $y_i$ in reverse from the input. The resulting accepting computation on $x^{(p)}$ would accept in $M'$. Conversely, any word $w$ accepted by $M'$ has a guessed sequence of states and a decomposition of the input such that $x^{(p)} =w$ where $M$ would accept $x$.
\end{proof}

%Note that the two lemmas above hold for the special case when
%$\NPCM$ and $\NCM$ are replaced by $\NPDA$ and $\NFA$, respectively.
Putting the previous two lemmas together, we obtain the characterization:

\begin{proposition} \label{prop20} 
The following are equivalent for a language $L$:
\begin{itemize}
\item
$L$ is accepted by a finite-turn $\NPDA$ (respectively, finite-turn $\NPCM$),
\item $L$ is generated by a $\ULG$ (respectively $\ULGC$),
\item there is a paired $2m$-tape $\NFA$ (respectively paired $2m$-tape $\NCM$) $M$ for some $m$, and a well-formed $m$-pattern
$p$ such that $L = L(M)^{(p)}$.
\end{itemize}
\end{proposition}

%When there are no counters, we have:
%\begin{corollary}
%$L$ is accepted by a finite-turn $\NPDA$ if and only if there is 
%a paired $2m$-tape $\NFA$ $M$ for some $m$ and a well-formed $m$-pattern $p \in P^m$ such that
%$L = L(M)^p$.
%\end{corollary}

This is a generalization of a known result \cite{lineargrammars} that
$L$ is accepted by a 1-turn $\NPDA$ (equivalently, generated by a linear
$\CFG$) if and only if there is $2$-tape $\NFA$ $M$ such that
$L = \{xy^R ~|~ (x,y) \in L(M) \}$; this only needs the pattern $\texttt{[ ]}$.

\section{Complexity of (Finite-Turn) $\NPCM$ Languages}
\label{sec:complexity}

We first consider the complexity of finite-turn $\NPCM$.
We will use the characterization of finite-turn $\NPCM$ in terms of
multi-tape $\NCM$ to show the following:

\begin{proposition} \label{BB4}
Let $L$ be a language accepted by a finite-turn $\NPCM$.  Then
$L$ is in $\NLOG$.
\end{proposition}
\begin{proof}
By Proposition \ref{prop20}, 
$L$ is accepted by a finite-turn $\NPCM$ if and only if there is 
a paired $2m$-tape $\NCM$ $M$ for some $m$ and a $m$-pattern $p$
such that
$L = L(M)^{(p)}$.

It is known that a (1-tape) $\NCM$ accepts in linear time, i.e.,
there is a constant $c$ such that every  input of length $n$ that is accepted
can be accepted in a computation that runs within $cn$ time (the number of transitions used in the run) \cite{Baker1974}.
This generalizes to any $t$-tape $\NCM$ $M$.  To see this, from $M$ construct
a (1-tape) $\NCM$ $M'$ which simulates the computation of $M$  
on $t$ tapes by using distinct symbols to correspond to the symbols
on the tapes; thus reading symbol $a_i$ will correspond to reading symbol $a$ on tape $i$. The machine $M'$ therefore
reads a shuffle (with added subscripts) of the $t$-tape words accepted by $M$.
Since $M'$ is a (1-tape) $\NCM$, it runs in linear time. As
$M'$ exactly simulates the computation of $M$ on the $t$ tapes,
it follows that $M'$ also operates in linear time.

Hence, the $2m$-tape $\NCM$ $M$ that characterizes the finite-turn
$\NPCM$ accepts in linear time.  We construct from $M$ a
$\log n$ space-bounded $\NTM$
$M'$ which accepts $L$ as follows.  
Since $M$ accepts in linear time, $M'$ needs only
a $\log n$ tape to simulate each counter. 
Now given an input $w$ and the pattern $p$, $M'$ first nondeterministically
guesses a decomposition of (the two-way input) $w$ into $2m$ segments
according to pattern $p$.  $M'$ needs $2m$ additional read/write tapes each of size $\log n$ to remember the
locations of the borders between segments.  It also needs $2t$ additional 
read/write tapes each of size $\log n$
to remember the locations of the $2m$ read heads within the
$2m$ tapes of $M$. Then $M'$ simulates the computation of $M$
on the $2m$ tapes.  Even though $M'$ has only one two-way read-only
input tape, it can simulate the computation of $M$
since the locations of the starting and end points of 
each $x_i$ and $y_i$ (some of which are reversed according
to pattern $p$) as well as the locations of the tape
heads of $M$ within the segments are remembered in the
read/write tapes.  Hence, $L$ is $\NLOG$.
\end{proof}

Next, we turn to the general model of $\NPCM$, which need
not be finite-turn.
It is known that every $\CFL$ (a language accepted by an $\NPDA$)
is in $\DSPACE(\log^2 n)$ \cite{CFLslogsquared,HU1969}.  We will show that
every $\NPCM$ language is also in 
$\DSPACE(\log^2 n)$.  %We will again use the equivalent $\mathbb{Z}$-grammars.

For the next result, it is convenient to use another grammar model that is also known to be
equivalent to $\NPCM$ \cite{Georg}, but has a slightly different definition than $\CFGC$.
A $k$-dimensional $\mathbb{Z}$-grammar is a tuple $G = (V, \Sigma, P, S)$,
where $V$ and $\Sigma$ are non-empty sets of nonterminals 
and terminals, $S \in V$ is the start nonterminal, and $P$ is a set of rules of the form:
$$A \rightarrow (\alpha, c_1, \ldots, c_k),$$
where $A \in V$ and $\alpha \in (V\cup\Sigma)^*$, and $c_1, \ldots, c_k$
are in $\mathbb{Z}$ (i.e., they can be positive, $0$, or negative integers).
A configuration of $G$ is a tuple $(\alpha, c_1, \ldots, c_k)$
where $\alpha \in (V\cup \Sigma)^*$, and
$c_1, \ldots, c_k \in \mathbb{Z}$.
We write 
$(\alpha A \beta, c_1, \ldots, c_k) \Rightarrow (\alpha \gamma \beta, d_1, \ldots, d_k)$, $A \in V, \alpha,\beta \in (V \cup \Sigma)^*$, if
there a rule $A \rightarrow (\gamma, a_1, \ldots, a_k)$ in $G$, and
$d_i = c_i + a_i$ for each $i$.
The reflexive, transitive closure of the relation $\Rightarrow$  is denoted by $\Rightarrow^*$.

A terminal string $w \in L(G)$ if there is a derivation
$(S, 0, \ldots, 0)  \Rightarrow^* (w, 0, \ldots, 0)$.
(Thus the counters are initially zero and the derivation of
a terminal string is successful if the counters also end 
with zero.)

It was shown in \cite{valence} that every $\mathbb{Z}$-grammar $G$
can be converted to an equivalent $\mathbb{Z}$-grammar $G'$ in Chomsky 
Normal Form ($\CNF$), i.e., the rules are of the form:
$A \rightarrow (BC, c_1, \ldots, c_k)$, or $A \rightarrow (a, c_1, \ldots, c_k)$,
where $A, B, C$ are nonterminal symbols and $a$ is a terminal
symbol such that that $L(G') = L(G)$ if $L(G)$ does not
contain $\lambda$ and $L(G') = L(G) - \{\lambda\}$ if $L(G)$ contains
$\lambda$. %We can define a $\CFGC$ in CNF in a similar way.

The following shows that $\mathbb{Z}$-grammars and $\CFGC$s are equivalent:
\begin{proposition}
Every $\mathbb{Z}$-grammar $G$ can be converted to an equivalent
$\CFGC$  $G'$, and conversely (hence, both define
the family of $\NPCM$ languages). Moreover, $G'$ is in $\CNF$ if and 
only if $G$ is in $\CNF$.
\end{proposition}
\begin{proof}
Given a $\mathbb{Z}$-grammar $G$ with $k$ counters $C_1, \ldots, C_k$,
we can construct a $\CFGC$ $G'$ with counters $C_1,\ldots, C_k$
and $D_1, \ldots,D_k$ with rules defined as follows:

If $A \rightarrow (\alpha, c_1, \ldots, c_k)$ is a rule in $G$ then
$A \rightarrow (\alpha, a_1, b_1, \ldots, a_k, b_k)$ is a rule in $G'$ where
$a_i = c_i$ and $b_i = 0$ if $c_i \ge 0$, and $a_i = 0$ and $b_i = c_i$
if $c_i < 0$.

By reversing the construction above, the converse is 
 also true.
Note that normality is preserved, i.e., 
$G'$ is in $\CNF$ if and only if $G$ is in
$\CNF$.
\end{proof}

The proof of the following result is the same as the proof
for the case of $\CFG$ in $\CNF$ in \cite{CFLslogsquared,HU1969}. Although we could prove it using $\CFGC$, it is more convenient to use $\mathbb{Z}$-grammars.

\begin{lemma} \label{techlemma}
Let $G=(V,\Sigma,P,S)$ be a $\mathbb{Z}$-grammar in $\CNF$.  
Suppose $(\alpha, c_1, ..., c_k )$ is a configuration,
and $(X, 0, \ldots, 0)  \Rightarrow^* (\alpha, c_1, \ldots, c_k)$ for $X \in V$ and $n =|\alpha| \ge 4$. Then $\alpha$ can be written as $\alpha_1 \alpha_2 \alpha_3$ with $\alpha_1, \alpha_3 \in (V \cup \Sigma)^*$ (possibly $\lambda$) and there are 
$a_1, \ldots, a_k, b_1, \ldots, b_k \in \mathbb{Z}$ such that there exists the following derivation 
of $(\alpha, c_1, ..., c_k)$:
\begin{itemize}
\item
$(X ,0,\ldots, 0) \Rightarrow^* (\alpha_1 A \alpha_3, a_1, \ldots, a_k) \Rightarrow^* (\alpha_1 \alpha_2 \alpha_3, c_1, \ldots, c_k)$,
\item
$(A, 0,\ldots, 0) \Rightarrow^* (\alpha_2, b_1, \ldots, b_k)$,
\item
$c_i = a_i + b_i$ for each $i$, and
\item
$1/3n < |\alpha_2| \le 2/3n$.
\end{itemize}
\end{lemma}

The following procedure and proof are similar to the corresponding proof for context-free languages  in \cite{CFLslogsquared,HU1969}.
\begin{proposition}
If $G$ is a $\mathbb{Z}$-grammar in $\CNF$, then $L(G)$ is in $\DSPACE(\log^2 n )$. Thus, every $\NPCM$ language is in $\DSPACE(\log^2 n)$.
\label{spaceNPCM}
\end{proposition}
\begin{proof}
First we observe that since $G$ is in $\CNF$,
there is a constant $m$ such that if $w$ of length $n$ is in 
$L(G)$ and $(S, 0,\ldots,0)  \Rightarrow^* (w, x_1, \ldots , x_k)$, then
the absolute value of any counter in every configuration encountered during the derivation is at most $mn$.

Because of the above lemma, we can write a recursive 
procedure to determine if a terminal string $w$ is in $L(G)$.  
When the recursive procedure is converted to a sequential procedure, the space requirement would be $\log^2 n$ space.
The recursive procedure is a generalization of the procedure
described in \cite{CFLslogsquared,HU1969}.

{\trivlist \small \item[\hskip \labelsep{\bf procedure} TEST($X,\alpha, c_1, \ldots, c_k$)]
\textcolor{white}{WTF}
\begin{code}
\noindent
\t\t // $X \in V$, $\alpha \in (V \cup \Sigma)^*$, and $c_i \in \mathbb{Z}$ with $|c_i| \le m|\alpha|$ for each $i$.\\  
\t\t // returns true if and only if
    $(X, 0, ..., 0)  \Rightarrow^* (\alpha, c_1, \ldots, c_k)$\\ 
\t\t if $|\alpha| \le 3$ and $(X, 0, \ldots, 0)  \Rightarrow^* (\alpha, c_1, \ldots, c_k)$ in at most  5 steps\\
\t\t\t return `yes'\\
\t\t else if there exist $\alpha_1, \alpha_2, \alpha_3 \in (V\cup \Sigma)^*$, $A \in V$, $a_1, \ldots, a_k, b_1, \ldots, b_k \in \mathbb{Z}$ such that \\
\t\t\t\t\t (1) $1/3|\alpha| < |\alpha_2| \le 2/3|\alpha|$\\
\t\t\t\t\t (2) $\alpha = \alpha_1 \alpha_2 \alpha_3$\\
\t\t\t\t\t (3) $ |a_i| , |b_i|, |c_i| \le m|\alpha|$ for $1 \le i \le k$\\
\t\t\t\t\t (4) $c_1 = a_1 + b_1, \ldots, c_k = a_k + b_k$\\
\t\t\t\t\t (5) TEST($X$, $\alpha_1 A \alpha_3$, $a_1, \ldots, a_k)$ and   TEST($A$, $\alpha_2, b_1, \ldots, b_k$)\\
\t\t\t return `yes'\\
\t\t else return `no'
\end{code}
 \endtrivlist}
To determine if a terminal string $w = a_1 \cdots a_n$
is in $L(G)$, we initially call TEST($S, a_1 \cdots a_n, 0, \ldots , 0)$.

Note that procedure TEST systematically looks at all possible decompositions of $\alpha$ into $\alpha_1, \alpha_2, \alpha_3$ 
and all choices of $a_1, \ldots, a_k, b_1, \ldots, b_k$ satisfying the constraints of (1) to (4) to check if one of the
decompositions of $\alpha$ combined with one choice of every $a_i$ and $b_i$, $1 \le i \le k$ satisfies (5). (If none do, then Lemma \ref{techlemma} implies the word is not accepted.) To systematically try all such decompositions and counter values involves having a loop that traverses in a depth-first brute-force manner through a vector
of $2k+2$ components, where the first $2k$ components are the current values of $a_1, \ldots, a_k, b_1, \ldots, b_k$ and
the remaining $2$ components are the starting and ending positions of $\alpha$. There are $2m|\alpha| + 1 $ values for each
of the counters (each can be negative) and $2/3 |\alpha|$ for each of the last two. This requires $m' \log |\alpha|$ bits for
the vector where $m'$ is a constant that only depends on $M$.

To see that the stated complexity holds, we  can convert the recursive procedure TEST to
a sequential procedure by ``unrolling'' the recursion.
The sequential procedure would require the use of a 
 stack. 
 Each stack
frame will need to record certain items such as $A$ and the vector with binary components. As justified with $m'$,
a stack frame requires only $O(\log n)$ space. 
Furthermore, in the recursion tree, each value of $n$ is at most $3/2$ times its children's value, and therefore
 the depth of recursion is $O(\log n)$, and
the stack will need $O(\log n)$ stack frames.  Hence, the
stack needs at most $O(\log n \log n) = O(\log^2 n)$
space.  
\end{proof}

Finally, we look at the time complexity of $\NPCM$ languages generally.
We will need the following
proposition from \cite{TatHungChan}:

\begin{proposition} \cite{TatHungChan} \label{BB1} 
Let $M$ be an $\NPCM$.  There is a constant $c$ such that
every $w$ in $L(M)$ has an accepting computation
that runs within $c|w|$ time. Thus, $M$ accepts in linear time.
\end{proposition}

The following is straightforward to show from results in the current
literature.

\begin{proposition} \label{rough} All $\NPCM$ languages are in $\P$.
\end{proposition}
\begin{proof}
In \cite{Cook:1971:Pushdowns}, it was shown that
nondeterministic Turing machines with a two-way read-only input tape, a $\log n$-space-bounded
read/write worktape, and a pushdown stack are equivalent to polynomial time
deterministic Turing machines.
By Proposition \ref{BB1}, every $\NPCM$ operates in linear time.
Hence the counter values are linear in the length of the input and can be
stored in $\log n$ space. Each such machine can be simulated by a nondeterministic Turing machine with a
$\log n$-space bounded read/write worktape plus a pushdown stack. Hence, it
follows that $\NPCM$ languages are in $\P$.
% and that all linear time $2\NPCM$ languages are in $\P$.
\end{proof}

\section{Two-way NPCM} \label{sec:2NPCM}

In this section, we will consider two-way $\NPCM$ machines, denoted by $2\NPCM$, which
have a two-way
read-only input with left and the right end markers $\rhd$ and $\lhd$. We assume that for two-way machines, acceptance occurs when the
machine falls off the right end marker in an accepting state. See \cite{Ibarra1978} for the formal
definition.

First, we see the following generalization from $\NPCM$ to $2\NPCM$.
\begin{proposition} All $2\NPCM$ where there is a polynomial $p(n)$ such that the counters grow to
at most $p(n)$ on inputs of size $n$ accept only languages in $\P$.
\end{proposition}
\begin{proof}
This is implied from the proof of Proposition \ref{rough}.
If there is a polynomial that bounds the size of the
counters, then it can be simulated by a two-way nondeterministic Turing machine with a $\log n$
space-bounded worktape to simulate each counter, and a pushdown, and the language it accepts is therefore
in $\P$.
\end{proof}

For $2\NPCM$, we will consider finite-turn machines as we did with one-way machines, where it makes at most $t$ turns on the pushdown, for some $t \ge 0$.
Next, we study bounds on the number of changes of direction on the two-way input tape.
We say a function $R(n)$ is {\em logspace-constructible} if there is a $\DTM$ that on an input of length $n$, stores
a string of  length at least $\log_2 R(n) +1$ on the worktape using at most $\log R(n)$ space (so it uses space to within a constant factor of the final string written). This string is long enough to be used as a binary counter up to at least $R(n)$. For a function  $R(n)$,
we say a $2\NPCM$ is
{\em $R(n)$-reversal} if every input of length $n$ accepted can be accepted where the input head makes at most $R(n)$ changes of direction on the two-way input tape.  We will assume that $R(n) \ge 1$ (note that when $R(n)= 0$, the machine is one way).
We call a $2\NPCM$ {\em sweeping} if the input head reverses occur only at the end markers.
We call a $2\NPCM$ {\em finite-reversal} if $R(n)$ is some constant, and we use
poly-reversal to mean $R(n)$ is $n^k$ for some constant $k$.

\begin{lemma} \label{tosweeping} Every $2\NPCM$ $M$ can be converted to an equivalent sweeping 
$2\NPCM$.
\end{lemma}
\begin{proof} 
The pushdown stack can be used to enforce the turns take place at the end markers. If the turn occurs at the end marker, there is no problem. If a turn occurs before the end marker it records on the top of the stack (using a different symbol) to count (by pushing) the number of moves until it reaches the end marker to which it was moving.  When it reaches the end  marker, it can reposition the head back to the location where it wanted to make a turn by popping the stack. Note, that if the number of reversals on the input is unbounded (finite), the number of added turns on the stack will be unbounded (finite). 
\end{proof}

In the proof of the next result, we introduce a technique that allows us
to use complexity results for one-way machines to prove complexity
results for two-way machines.
\begin{proposition} \label{iturn} Let $R(n)$ be a logspace-constructible function.
All $R(n)$-reversal $2\NPCM$ languages are
in $\DSPACE(\log^2 (nR(n)))$.
\end{proposition}
\begin{proof} 
Let $M$ be an $R(n)$-reversal $2\NPCM$.  From Lemma \ref{tosweeping}, we may assume
that $M$ is sweeping. Let $\#$ be a new symbol. 
Given a string $x_0 \in \Sigma^*$ and $k \ge 0$, let $x_0^{(k)} = x_0 \# \cdots \# x_k \#$ where $x_i = x_0$ for each even $i$, and 
$x_i = x_0^R$ for each odd $i$.
Construct a (one-way) $\NPCM$ $M'$ such that $M'$ accepts
strings of the form $w = x_0\#x_1\# \cdots \#x_k\#$, $k \ge 0$ even, such that:
$w = x_0^{(k)} \in L(M')$  for some $k, 0  \le k \le R(n)$
if and only if $M$ accepts $x_0$.
Certainly $M'$, on input $x_0\#x_1\# \cdots \#x_k\#$, $k \ge 0$ even (which has one-way input)
can simulate $M$ as if the input on the first sweep is $x_0$, on the second sweep is $x_1$, etc.\ until $x_k$. 
If $ x_i = x_0$ for each even $i$ and
$x_i = x_0^R$ for each odd $i$, then $M$ accepts $x_0$ as $M$ is sweeping, and conversely. 
Note that on the segment $x_0^R$ , $M'$ simulates the left moves of $M$ by moving
its (one-way) head to the right.    However, we note that if in $M$,  each counter
makes at most $r$ reversals, then the number of reversals each  counter in $M'$
also makes at most $r$ reversals, since the simulation is faithful.  Thus, the
number of counter reversals each counter of $M'$ makes does not increase.

Since $M'$ is an $\NPCM$, $L(M')$ is in $\DSPACE(\log^2 N)$ by Proposition \ref{spaceNPCM}, where
$N = |w|$. Note that when $x_i = x_0$ for each even $i$ and $x_i = x_0^R$
for each odd $i$, $N = (n + 1)(k + 1)$, where $n = |x_0|$.
Let $Z$ be one such deterministic Turing machine accepting $L(M')$. Note that $Z$ has a two-way input.
We then construct another $\DTM$ $Z'$ which when given a two-way input $x_0$, 
first puts a string $y$ of length at least $\log_2 R(n) +1$ on the worktape, which can be done in $\log R(n)$ space as it is logspace-constructible. It then
systematically simulates $Z$ to check if there is a $k$ starting at $0$ and increasing by $1$ repeatedly until $k$ is $2^{|y|} \ge R(n)$ (by using $y$ as a binary counter). Then  $Z$ accepts an input $w = x_0^{(k)}$ for some $k$, if and only if $Z'$ accepts $x_0$. Note that $Z'$ has only a two-way input $x_0$ but it has to simulate $Z$ on the two-way 
input $w$. If $Z'$ is to accept $w$, it will do so where $k$ is at most $R(n)$ and hence, if $Z'$ accepts, this can be done using $\log_2 R(n)+1$ space to keep track of the segment and symbol within the segment $Z$ is scanning.  Since $L(M) = L(Z')$,  it follows that each $R(n)$-reversal $2\NPCM$ language
is in $\DSPACE(\log^2 (nR(n)))$.
\end{proof}

\begin{remark}
In the proposition above, we used the standard definition of space complexity of Turing machines, as in 
``for every input of size $n$, the machine uses at most $f(n)$ space''  \cite{HU}. There is also space measured in the so-called
weak sense: ``for every accepted word, there is an accepting computation that uses at most $f(n)$ space'' \cite{Giovanni}.
Had we used the weak sense, we would not need to assume $R(n)$ is logspace-constructible and would just need to assume that the function exists in the proposition above; this is because in the construction of $Z'$, $k$ could be increased by 1 repeatedly until it accepts (which must be where $k$ is at most $R(n)$) and then it halts; if it does not accept, it never halts.
\end{remark}
In the next corollary, it is evident that $c^n$, for a constant $c$ is logspace-constructible as given $c$, we can put $\log_2 n^c = c \log_2 n$ on the worktape in $\log$ space. Also, $c^n$ is logspace-constructible as $\log_2 c^n = n \log_2 c$. Lastly,
$c^{\sqrt{n}}$ is logspace-constructible as  $\log_2 c^{\sqrt{n}} = \sqrt{n} \log_2 c$. To obtain a string of this length,
we can store $\log_2 n$ in one worktape and using additional worktapes, for each $i$ starting at 1 in increments of 1, calculate $i^2$ until it is longer than $\log_2 n$.

\begin{corollary} \label{cor1} 
The following are true:
\begin{enumerate}
\item
All poly-reversal $2\NPCM$ languages are
in $\DSPACE(\log^2 n)$.
\item
For any constant $c$, all $c^{\sqrt{n}}$-reversal $2\NPCM$ languages  are in $\DSPACE(n)$.
\item For any constant $c$, all $c^n$-reversal $2\NPCM$ languages are in $\DSPACE(n^2)$.
\end{enumerate}
\end{corollary}

Corollary \ref{cor1} item 1 generalizes the well-known result
that $\NPDA$ languages (context-free languages) are in
$\DSPACE(\log^2 n)$, and items 2 and 3 compare well
with the known result that $2\NPDA$ languages are
in $\DSPACE(n^2)$ \cite{AhoHopcroftUllman}.
We do not know if $2\NPCM$ languages (with no head reversal restriction)
are in $\DSPACE(n^2)$. If for every $2\NPCM$ there is a constant
$c$ such that every input of length $n$ can be accepted within at most
$c^n$ head reversals, then the languages are in $\DSPACE(n^2)$ as shown
in item (3) above.

But, as was mentioned in Section \ref{sec:intro}, it was recently shown that the membership problem
for $2\NPCM$ is $\NP$-complete and hence, in $\PSPACE$ (which is equal to deterministic polynomial space).

%It also implies that in Proposition \ref{rough} (2) the polynomial
%counter restriction is necessary to be in $\P$ (unless $\P = \NP$).

As we have seen in the proof of Lemma \ref{tosweeping}, when we convert a $2\NPCM$ to
an equivalent sweeping $2\NPCM$, the resulting $2\NPCM$ will have an unbounded
(finite) number of additional stack turns if the number of input head reversals is unbounded (finite). 
The proof of the next result is the same as that of Proposition \ref{iturn}, using
the fact that finite-turn $\NPCM$ is in $\NLOG = \NSPACE(\log n)$ (rather than $\DSPACE(\log^2 n)$ without the finite-turn restriction in the construction of $Z$ from $M'$) by Proposition \ref{BB4}.

\begin{proposition}  \label{nlog}
The following are true:
\begin{enumerate}
\item All languages accepted by finite-reversal
 finite-turn $2\NPCM$ are in $\NLOG$.
\item For any logspace-constructible function $R(n)$, all languages accepted by $R(n)$-reversal
 finite-turn sweeping $2\NPCM$ are in
$\NSPACE(\log (nR(n)))$.
\end{enumerate}
\end{proposition}

The result above is a generalization of the known result that
finite-turn $\NPDA$
(i.e., one-way $\NPDA$ whose stack makes at most
$k$ turns from some $k$) languages are in $\NLOG$ \cite{finiteturn}.
In particular, item 2 shows that
all languages accepted by poly-reversal
finite-turn sweeping $2\NPCM$ are in $\NLOG$ (without the sweeping restriction, this simulation causes the pushdown to no longer be finite-turn).

It was also shown in \cite{finiteturn}
that finite-turn $\DPDA$ (i.e., deterministic) languages
are in $\DLOG$ (= $\DSPACE(log~ n))$. Using this result and similar
techniques as above, we have (where $2\DPDA$ are two-way deterministic pushdown automata):

\begin{proposition}  \label{dlog}
The following are true:
\begin{enumerate}
\item All languages accepted by finite-reversal
 finite-turn $2\DPDA$ are in $\DLOG$.
\item For any logspace-constructible function $R(n)$, all languages accepted by $R(n)$-reversal
 finite-turn sweeping $2\DPDA$ are
 in $\DSPACE(\log (nR(n)))$.
\end{enumerate}
\end{proposition}

We do not know if Proposition \ref{dlog} is true when $2\DPDA$ is
replaced by $2\DPCM$, as we do not know at this time whether the languages
accepted by finite-turn $\DPCM$ are in $\DLOG$.

\begin{comment}
{\bf QUESTION for Oscar: I'm not sure what the result below is trying to say because the NPCM Z'' accepts a different language than the original machine -- so we cannot say the time complexity of the language is the same.. are we saying the complexity of the fixed membership problem instead?  But the Z'' machine is not solving the fixed membership problem for M.}
\end{comment}

Similarly, we note the following regarding time complexity. We say a function $T(n)$ is {\em time constructible} if a $\DTM$
can compute it in $T(n)$ time.
\begin{proposition} \label{time}
Let $T(n)$ be a time-constructible function such that all $\NPCM$ (respectively $\NPDA$) languages are in $\DTIME(T(n))$.
Then every finite-reversal $2\NPCM$ (respectively finite-reversal $2\NPDA$) language is in $\DTIME(T(n))$.
Further, all poly-reversal $2\NPCM$ (respectively poly-reversal $2\NPDA$) languages are in $\P$.
\end{proposition}
\begin{proof}
Let $M$ be $k$-reversal.
As in the proof of Proposition \ref{iturn}, let $Z$ be a $T(n)$-time-bounded deterministic Turing machine accepting $L(M')$.
% (created at the end of that proof).
We construct another $\DTM$ $Z'$ which starts by computing $T(n)$. Then, 
given input $x_0$, it writes on a read/write tape the string
$w = x_0^{(k)}$. Then $Z'$ simulates $Z$ on $w$ while stopping after $T(n)$ steps.  It is evident that 
$Z'$ has the same time complexity as $Z$.

For the second statement, if $M$ is $n^c$-reversal for some constant $c$, then following the same proof of Proposition \ref{iturn},
the one-way machine $M'$ accepts a language in $\P$ by Proposition \ref{rough}. So and so a Turing machine $Z$
can be built running in $O(n^k)$ time for some constant $k$ accepting $L(M')$. From $Z$, another Turing machine $Z'$ can be built that on input
$x_0$, puts $x_0^{(n^c)}$ on a worktape in $O(n \cdot n^c)$ time, and then simulates $Z$ that runs in time complexity $O(n^{(c+1)k})$.
\end{proof}

While typically membership for $\NPDA$ is decided on a random access machine (and is faster than $O(n^3)$), it runs in $n^3$ time on a $\DTM$ \cite{HU1969}.
\begin{corollary}
Every language accepted by a  finite-reversal $2\NPDA$ can be accepted by a $\DTM$ in $n^3$ time.
\end{corollary}

The two-way to one-way reduction technique in the proof of Proposition \ref{iturn}
can be used to prove other complexity results involving two-way machines.
For example, in \cite{ChanStoc}, Chan studied
$\NPDA$ and $2\NPDA$ with  counters that are $c(n)$-reversal-bounded (i.e., the
number of changes between sequences of non-increasing and non-decreasing  in  each counter is at most $c(n)$ on an input of size $n$), generalizing 
$\NPCM$ and $2\NPCM$, respectively. He showed the following, where $c(n)$ is time constructible:
\begin{enumerate}

\item An $\NPDA$ with $c(n)$-reversal-bounded counters can be simulated by an $\NTM$
(i.e., a nondeterministic multi-tape TM) within time  $p(n + c(n))$ for some
polynomial $p$.
\item
A $2\NPDA$ with $c(n)$-reversal-bounded counters can be simulated by an $\NTM$
within time $p(n^{n^2}+ c(n))$ for some  polynomial $p$.
\end{enumerate}
Chan asked whether the $\NTM$ simulation time above can be lowered. He noted
that in \cite{Gurari},  $2\NCM$ (i.e., $2\NPCM$ without reversal-bounded counters) accept only languages
in $\NLOG$
and hence in $\P$,  but from (2) above, the $\NTM$  simulation time for the same
machines augmented  with a pushdown store is still $p(n^{n^2}+ c)$ (constant $c$) for
some
polynomial $p$.  We will see that a better simulation time can be obtained
when
the number of turns the two-way input head makes is bounded.

Consider an $R(n)$-reversal $2\NPDA$ $M$  with $c(n)$-reversal-bounded counters.  Then, we  can use item (1) above
and the proof technique in Proposition \ref{iturn} to show the following:
\begin{proposition}  Let $R(n)$ and $c(n)$ be time constructible. An $R(n)$-reversal $2\NPDA$ $M$ with $c(n)$-reversal-bounded counters can be simulated by an $\NTM$ within time
$p(nR(n) + c(n))$ for some polynomial $p$.
\end{proposition}
The reason the term $c(n)$  above does not change is that in converting the
$2\NPDA$  $M$ to an $\NPDA$ $M'$ using the technique of Proposition \ref{iturn}, the
number of  reversals each counter  of $M'$ makes is  the same as those of $M$,
since  the simulation  of $M$ by $M'$ is faithful, as noted in the proof of
Proposition \ref{iturn}.

Thus, for $R(n)$-reversal $2\NPDA$, we
can get smaller time bounds than in (2) when $R(n)$ is small.  For example:
\begin{corollary} The following are true:
\begin{enumerate}

\item A poly-reversal $2\NPDA$ $M$ with polynomial reversal-bounded counters can be simulated by an $\NTM$
within time $p(n)$ for some polynomial $p$.
\item An $r^n$-reversal $2\NPDA$ $M$  with $c^n$-reversal-bounded counters (for
some constants $r$ and $c$) can be simulated by an $\NTM$ within time $d^n$ for
some constant $d$.
\end{enumerate}
\end{corollary}

\section{Finite-Flip $\NPCM$ (Finite-Flip $2\NPCM$)}

In this section, we will show that the space complexity results
we obtained for the machines in the previous sections 
generalize to when the machines can reverse its pushdown
a finite number of times during the computation. This generalization of $\NPDA$ has been previously studied \cite{flipPushdown,flipPushdown2}.

A flip-$\NPDA$ $M$ is an $\NPDA$ which has the ability to
flip its pushdown at various times during its computation
by entering one of a designated set of flip states. Thus,
when the pushdown contains $Z_0 \alpha$ (where $Z_0$ is the bottom-of-stack marker which is never replaced) and $M$ enters a 
flip state, the stack becomes $Z_0 \alpha^R$. A machine  $M$ is
$k$-flip (respectively finite-flip) if every accepted string can be accepted in a computation
where $M$ makes at most $k$ flips (respectively a finite number of flips) of the stack.  In fact,
without loss of generality, we will assume 
that every input is accepted with the machine using
exactly $k$ flips of the stack (instead of at most $k$
flips).  This is because the state can count the number
of flips the machine  has made and can perform additional dummy
flips before accepting to make the number of flips
exactly $k$. 
Note that a 0-flip $\NPDA$ is just an $\NPDA$.
We refer the reader to
\cite{flipPushdown,flipPushdown2} for formal definition.

We can then also define $k$-flip $\NPCM$, finite-flip $\NPCM$,
finite-flip $2\NPCM$, finite-turn finite-flip $2\NPCM$
(which is a finite-flip $\NPCM$ that makes a finite number
of turns on it pushdown between flips), $R(n)$-reversal finite-flip
$2\NPCM$ (which is a finite-flip $2\NPCM$ that makes at
most $R(n)$ input head reversals), etc. Thus the models in
the previous sections are generalized by replacing `$\NPCM$'
and `$2\NPCM$' by `finite-flip $\NPCM$' and `finite-flip $2\NPCM$',
respectively.

\begin{comment}

It has been shown in \cite{flipPushdown} that a language
is accepted by a $k$-flip $\NPDA$ by null stack if and only if
it is also accepted by a $k$-flip NPDA by accepting state
(thus the two notions of acceptance are equivalent).

\end{comment}

The next lemma follows from Theorem 5 in \cite{flipPushdown}, which does not use the $\$$ symbol, but it can be added.
\begin{lemma} \cite{flipPushdown}
 \label{npda}
Let $k \ge 0$. A language $L$ can be accepted by a ($k+1$)-flip $\NPDA$ $M = (Q, \Sigma, \Gamma, q_0, F, \delta)$ that makes exactly $k+1$ on every accepting computation if and only if 
the language 
$$L_R=\{ u \$ v^R \mid \begin{array}[t]{l} (q_0, u, Z_0 ) \vdash^* (q_1, \lambda, Z_0 \gamma) \mbox{~with~} k \mbox{~flips, there is a flip from~} q_1 \mbox{~to~} q_2, \mbox{~and} \\
(q_2,v, Z_0 \gamma^R) \vdash^* (q_3, \lambda, Z_0), q_3 \in F  \},\end{array}$$ 
is accepted by a $k$-flip $\NPDA$ $M'$.
Moreover, if $M$
is finite-turn, then $M'$ is also finite-turn.
\end{lemma}

The above result holds when the $\NPDA$ is replaced by
$\NPCM$, as we show next.
\begin{lemma} \label{npcm}
Let $k \ge 0$. A language $L$ can be accepted by a ($k+1$)-flip $\NPCM$ $M = (Q, \Sigma, \Gamma, q_0, F, \delta)$ that makes exactly $k+1$ on every accepting computation if and only if 
the language 
$$L_R=\{ u \$ v^R \mid \begin{array}[t]{l} (q_0, u, Z_0 ) \vdash^* (q_1, \lambda, Z_0 \gamma) \mbox{~with~} k \mbox{~flips, there is a flip from~} q_1 \mbox{~to~} q_2, \mbox{~and} \\
(q_2,v, Z_0 \gamma^R) \vdash^* (q_3, \lambda, Z_0), q_3 \in F  \},\end{array}$$ 
is accepted by a $k$-flip $\NPCM$ $M'$.
Moreover, if $M$
is finite-turn, then $M'$ is also finite-turn.
\end{lemma}
\begin{proof}
Let $\Sigma$ be the input alphabet of $M$, and $C_1,\ldots, C_m$ be
its monotonic counters, where $m$ is an even positive integer.  Construct an $\NPDA$ $M_1$ with input alphabet 
$\Delta = \Sigma \cup \{c_1, \ldots , c_m\}$, where the $c_i$'s are new 
symbols corresponding to the counters. When given an input
$x \in \Delta^*$, $M_1$ simulates the computation of $M$ but reads
symbol $c_i$ whenever $M$ increments counter $C_i$.  

By Lemma \ref{npda}, we can construct a flip 
$\NPDA$ $M_1'$ which makes exactly $k$ flips that accepts $L(M_1)_R$. Note that the input alphabet
of $M_1'$ is $\Delta$.  We can then construct a flip $\NPCM$ $M'$ with counters $C_1, \ldots, C_m$ to simulate $M_1'$, but with input $\Sigma$. So in the simulation, $M'$ instead
of reading a letter $c_i$, increments counter $C_i$. In this way, the places on
the input where it is supposed to read $c_i$ are nondeterministically guessed.
Then $M'$ accepts $L_R$.
\end{proof}

\begin{proposition} \label{flip1}
Every language accepted by a finite-flip $\NPCM$ is in
$\DSPACE(\log^2 n)$.
\end{proposition}
\begin{proof}
The construction is by induction on the number of flips $k$. We know already that $0$-flip $\NPCM$ are equal to $\NPCM$, which are in 
$\DSPACE(\log^2 n)$ by Proposition \ref{spaceNPCM}.
Assume by induction that all $k$-flip $\NPCM$ are in $\DSPACE(\log^2 n)$.
Let $M$ be an $\NPCM$ which makes exactly $k+1$ flips.
From Proposition \ref{npcm}, we can construct a flip $\NPCM$ $M'$ which makes exactly $k$ flips such that 
$L(M') = L(M)_R$.
By induction, $L(M')$ can be accepted by a $\DTM$ $Z'$ in $\log^2 n$ space. 

We construct a  $\DTM$ $Z$ to accept $L(M)$ which, when given an input $w$ 
systematically splits the input into $u$ and $v$ such that $w = uv$ in every possible way, starting with $|u| = 0$.  Then $Z$ simulates the $\DTM$ $Z'$  on $u\$v^R$ (it needs only $\log~ n$ space to remember the location of the split).  Clearly, the simulating $\DTM$ $Z$ is also $\log^2 n$ space-bounded. 

We will prove $L(Z) = L(M)$.
First 
Let $w \in L(M)$. Then $w = u v$ where $u$ is the part that is read up until the $k^{\mbox{th}}$ flip, it flips the final time, then it reads $v$ and accepts. Because of that,
$u \$ v^R$ is in $L(M)_R = L(Z')$. Then in $Z$, on input $w$, it will eventually try the partition $w = uv$ and check that $u\$ v^R$ is in $L(Z')$ and then accept $w$. Conversely,  let $w \in L(Z)$. Then on input $w$, there must be some split $w = uv$ such that $u\$v^R \in L(M)_R$. By the construction of $L(M)_R$, this means $uv$ accepts in $M$.
\end{proof}

By the same ideas as above, we can also show the following:

\begin{proposition} \label{flip2}
Every language accepted by a finite-turn finite-flip $\NPCM$ is in
$\NLOG$.
\end{proposition}

Using Propositions \ref{flip1} and \ref{flip2}, we can then
show (using the same proof techniques) that all the results
in Sections 4 and 5 hold when `$\NPCM$' and `$2\NPCM$' are
replaced by `finite-flip $\NPCM$' and `finite-flip $2\NPCM$' respectively.
In particular, we have:

\begin{proposition} \label{spaceflip}
The following are true:
\begin{enumerate}
\item For logspace-constructible $R(n)$,
every language accepted by an $R(n)$-reversal finite-flip $2\NPCM$
is in $\DSPACE(\log^2(nR(n)))$.
\item
Every language accepted by a poly-reversal finite-flip $2\NPCM$
is in $\DSPACE(\log^2 n)$.
\item
Every language accepted by a finite-reversal finite-turn finite-flip $2\NPCM$
is in $\NLOG$.
\end{enumerate}
\end{proposition}

For the time complexity of finite-flip $\NPCM$ languages, we have:
\begin{proposition} Let $T(n)$ be a time-constructible function. If every $\NPCM$ language is accepted by a $\DTM$ in
$T(n)$ time, then every $k$-flip $\NPCM$ language can be accepted 
by a $\DTM$ in $n^kT(n)$ time.
\end{proposition}
\begin{proof} The proof is again an induction on $k$. This is obviously
true for $k = 0$.  Now assume that every $k$-flip $\NPCM$ language
can be accepted by a $\DTM$ in $n^kT(n)$ time. Let $M$ be $(k+1)$-flip $\NPCM$. We construct a $k$-flip $\NPCM$ $M'$ as in the proof of Proposition \ref{flip1}. 
Then, by the induction hypothesis, $L(M')$ can be accepted by a $\DTM$
$Z$ in $n^kT(n)$ time. From $Z$, we construct a $\DTM$ $Z'$ to accept
$L(M)$ as in that proof. Since there are $n$ possible positions in 
the input $w$ to make the split, systematically checking each of these for a ``valid'' split will add a factor of $n$ in the  complexity.  It follows that $Z'$ will have time complexity 
$n^{k+1}T(n)$.
\end{proof}
Since we know that every $\NPCM$ language is in $\P$ 
by Proposition \ref{rough}, we have:
\begin{corollary} Every language accepted by a finite-flip $\NPCM$ is in $\P$.
\end{corollary} 
We can then generalize the above result (using the two-way to one-way conversion technique) to:
\begin{corollary} Every language accepted by a poly-reversal finite-flip $2\NPCM$ is in $\P$. \label{polyreversalflip}
\end{corollary}
Clearly, the proposition above holds for $k$-flip NPDA languages.  Since every $\NPDA$ language ($\CFL$) can be accepted by a $\DTM$ in 
$O(n^{3})$ time \cite{HU1969}, we have:
\begin{corollary}  Every language accepted by a $k$-flip $\NPDA$ can be 
accepted by a $\DTM$ in $n^{k+3}$ time.
\end{corollary}
This generalizes to:
\begin{corollary} Every language accepted by a finite-reversal
$k$-flip $2\NPDA$ can be accepted by a $\DTM$ in $n^{k+3}$ time.
\end{corollary}

\section{Conclusions}
In this paper, we provided characterizations of one-way nondeterministic finite-turn pushdown automata and finite-turn pushdown automata augmented with reversal-bounded counters ($\NPCM$). The first characterization uses a restriction of multi-tape nondeterministic finite automata, and the second uses multi-tape reversal-bounded multi-counter machines. This result generalizes a known characterization for linear context-free grammars.

We then studied the complexity of the languages accepted by $\NPCM$. We showed that
finite-turn $\NPCM$ languages are in $\NLOG$. For (unrestricted) $\NPCM$, the languages they accept are in $\P$, and they also can be  accepted by $\log^2 n$-space-bounded deterministic Turing machines, matching the space complexity of
the languages accepted by pushdown automata. $\NPCM$ is an interesting class as it accepts more languages than both reversal-bounded multicounter machines and pushdown automata, but also has a $\coNP$-complete boundedness and emptiness problems.
We then studied extensions to one- and two-way $\NPCM$, where the pushdown can flip a finite number of times.

We summarize the complexity results for the models studied in this paper within Table \ref{complexitytable}.
In \cite{decisionProblemsNCM}, the complexity of some of the same machine models studied here was also studied. The decision problems that were analyzed included the non-emptiness problem and the general membership problem. The general membership problem (sometimes called the uniform membership problem in the literature) is, ``given both the machine $M$ and $w$ as inputs, determine whether $M$ accepts $w$?''. In \cite{decisionProblemsNCM}, 
 it was shown that for any type of machine model, the complexity of general membership for the two-way machine model is polynomial-time equivalent to the complexity of general membership for the one-way machine model which is in turn polynomial-time equivalent to the non-emptiness problem for the one-way machine model. This allows us to fill in the complexity for the general membership problem for most machine models studied by using the complexity for the non-emptiness problem. Since the complexity of the general membership problem for both $\NCM$ and $\NPCM$ are $\NP$-complete \cite{decisionProblemsNCM}, it follows that it must be also for finite-turn $\NPCM$ as well. 
The problem is open for finite-flip $\NPDA$ and its two-way restrictions.
 For the space complexity results, it is known that if all $\NPDA$ languages are in $\DSPACE(\log^2 n)$ \cite{CFLslogsquared}; further, they are all in $\DLOG$ if and only if $\NLOG = \DLOG$ \cite{SUDBOROUGH}; certainly all $\NPCM$ being in $\DLOG$ would also imply $\NLOG = \DLOG$. 
\begin{table}[htp]
\caption{For some different machine models (in column 1), we provide the time complexity (column 2) and space complexity (column 3) of languages accepted by the models, and the complexity of the general membership problem (column 4). The previous paper or result in this paper that provides a proof is given beside each entry.}
\begin{center}
{\footnotesize
\begin{tabular}{|c||c|c|c|}\hline
{\bf Model}				& time & space  & general membership 	\\\hline
$\NPDA$					& $\P$ \cite{HU}	& $\DSPACE(\log^2 n)$ \cite{CFLslogsquared} & $\P$ \cite{GeneralMembership} \\\hline
finite-turn $\NPDA$					& $\P$ \cite{HU} & $\NLOG$ \cite{finiteturn} & $\P$ \cite{GeneralMembership}  \\\hline
$\NCM$					& $\P$ \cite{decisionProblemsNCM}	& $\NLOG$ \cite{decisionProblemsNCM} & $\NP$-complete \cite{HagueLin2011} \\\hline
$2\NCM$					& $\P$ \cite{Gurari}	& $\NLOG$ \cite{Gurari} & $\NP$-complete \cite{decisionProblemsNCM} \\\hline
finite-turn $\NPCM$					& $\P$ Prop \ref{BB4}&  $\NLOG$ Prop \ref{BB4} & $\NP$-complete \cite{decisionProblemsNCM} \\\hline
$\NPCM$					& $\P$ Prop \ref{rough}	& $\DSPACE(\log^2 n)$ Prop \ref{spaceNPCM} & $\NP$-complete \cite{HagueLin2011} \\\hline
poly-reversal & & & \\ $2\NPCM$ & $\P$ Prop \ref{time} & $\DSPACE(\log^2 n))$  Cor \ref{cor1} & $\NP$-complete \cite{HagueLin2011,decisionProblemsNCM}  \\\hline
$2\NPCM$					& $\NP$-complete \cite{decisionProblemsNCM} & $\PSPACE$ \cite{decisionProblemsNCM} & $\NP$-complete \cite{decisionProblemsNCM}\\\hline
poly-reversal & & & \\ finite-flip $2\NPCM$ & $\P$ Corollary \ref{polyreversalflip} & $\DSPACE(\log^2 n))$  Prop \ref{spaceflip} & ?  \\\hline
\end{tabular}}
\end{center}
\label{complexitytable}
\end{table}%

For future directions, although we know that all $\NPCM$ languages are in $\P$, we do not know if we can decide whether or not a word of length $n$ is accepted in $O(n^k)$ for some fixed $k$, as is the case with $\NPDA$ \cite{HU}. This would be especially important towards their practical utility. Also we do not know a good upper bound on the space needed to accept $2\NPCM$ languages. We know
that it is in $\PSPACE$ and if the input head reversals are restricted 
to $c^n$ (where $n$ is the length of the input, $c$ a constant), the
languages are in $\DSPACE(n^2)$ which matches the space upper bound
for languages accepted by $2\NPDA$ (i.e., $2\NPCM$ without the
reversal-bounded counters). 
Also, the time complexity of general membership for finite-flip $\NPDA$ (and its extensions to two-way machines) requires further study.

\section*{Acknowledgements}
Supported, in part, by Natural Sciences and Engineering Research Council of Canada Grant 2022-05092 (Ian McQuillan).
We also thank the reviewers for the thoughtful reading of the manuscript.

\bibliography{bounded}{}
\bibliographystyle{ws-ijfcs}

\end{document}